\documentclass[12pt,reqno]{amsart}
\usepackage{fullpage}
\usepackage{amsfonts}
\usepackage{amssymb}
\usepackage{times}
\usepackage{graphicx}
\newtheorem{thm}{Theorem}[section]
\newtheorem{cor}[thm]{Corollary}
\newtheorem{lem}[thm]{Lemma}

\newtheorem{prop}[thm]{Proposition}

{ \theoremstyle{remark}\newtheorem{remark}{Remark} }

\usepackage[linkcolor=blue,colorlinks=true]{hyperref}
\usepackage{amsfonts,rawfonts}
\usepackage{todonotes}
\usepackage{bbold}
\usepackage{thmtools}
\usepackage{systeme}
\usepackage{mathtools}
\usepackage{xcolor}
\usepackage{xfrac}
\usepackage{physics}
\usepackage{enumitem}
\usepackage{mathtools} 
\usepackage{tikz-cd}
\usetikzlibrary{arrows,chains,matrix,positioning,scopes}
\usepackage{bbm}

\usepackage{comment}
\usepackage{enumitem}

\newtheorem{definition}{Definition}

\newcommand{\Z}{\mathbb Z}

\newcommand{\Q}{\mathbb Q}
\newcommand{\R}{\mathbb R}

\newcommand{\hZ}{\hat{\mathbb{Z}} }

\def \d{\mbox{\(\,\mathrm{d}\)}}
\usepackage{mathtools}

\begin{document}

\title[Short version of title]{ Bost-Connes-Marcolli system for the Siegel modular variety}
\author{Ismail Abouamal}
\email{abouamal@caltech.edu}
\maketitle

\begin{abstract}
We  present a generalization of the Connes-Marcolli $GL_{\mathbb Q,2}$-system by constructing a quantum statistical mechanical system. Specifically, we introduce the Connes-Marcolli system associated with the Siegel modular variety of degree $2$. We investigate the system's $\textmd{KMS}\beta$-states for various inverse temperatures $\beta>0$. Our results reveal a spontaneous phase transition occurring at $\beta=3$. We demonstrate that the system lacks a $\textmd{KMS}_\beta$ state for $\beta<3$ with $\beta \neq 1$, identify the explicit extremal Gibbs states for $\beta>4$, and prove that a unique $\textmd{KMS}_\beta$ state exists for every $\beta>0$ with $3<\beta\leq 4$.
\end{abstract}


\section{Introduction}

By the end of last century, Bost and Connes \cite{bost1995hecke}, motivated by the ideas introduced by B. Julia \cite{julia1990statistical}, constructed a quantum statistical mechanical system $(\mathcal{A},(\sigma_t)_{t\in \mathbb R})$ with unexpected  connection between class field theory of $\mathbb Q$ and the theory of $C^*$-dynamical systems. Bost and Connes defined a dense rational subalgebra $\mathcal{A}_{\mathbb Q} \subset \mathcal{A}$ such that the evaluation of equilibrium states at low temperatures on $\mathcal{A}_{\mathbb Q}$ generate the maximal abelian extension $\mathbb Q^{\textmd{ab}}$ of $\mathbb Q$. As a quantum statistical $C^*$-dynamical system, the Bost-Connes system has also the following interesting  thermodynamical property: it exhibits a phase transition at inverse temperature $\beta =1$. Moreover, this transition happens to be spontaneous in the sense that the symmetry of the system $(\mathcal{A},(\sigma_t)_{t\in \mathbb R})$ changes radically with small changes of temperature. For $\beta >1$, the system admits $\zeta(\beta)$ as its partition function and the extremal $\textmd{KMS}_\beta$ states have type $\textmd{I}$ while for $0 < \beta \leq 1$ it admits a unique $\textmd{KMS}_\beta$ states of type $\textmd{III}_1$.

Over the past several years, several generalizations of the BC system have been studied. The construction of Bost and Connes was first generalized by Connes and Marcolli \cite{connes2004physics} to quadratic number fields. They introduced the so called the $GL_{2,\mathbb Q}$-system and showed that maximal abelian extensions of quadratic number fields are generated by evaluating the ground states on a dense rational arithmetic subalgebra. In a subsequent work by Connes, Marcolli and Ramachandran \cite{connes2006kms}, the connection between the $GL_2$-system and CM-fields was formally studied. The construction of Connes and Marcolli was further generalized to arbitrary number fields by Ha and Paugam \cite{ha2005bost}. The former authors reformulated the $GL_{2,\mathbb Q}$ system in the adelic language, making explicit its relation to the Shimura datum $(GL_2,\mathbb H^{\pm})$ over any given number field. They generalized the construction of Bost and Marcolli to an arbitrary Shimura datum $(G,X)$ and introduced a formal definition of the abstract Bost-Connes-Marcolli system associated to the pair $(G,X)$. It was shown that these systems admit the Dedekind zeta function as the partition function and the group of connected components of the id\`ele class group acts as the symmetry group.

In \cite{laca2007phase}, Laca, Larsen and Neshveyev gave yet another reformulation of the $GL_{2,\mathbb Q}$-system in terms of groupoid $C^*$-algebras. They recovered the classification results of $\textmd{KMS}_\beta$ states obtained by Connes and Marcolli and proved uniqueness in the critical range  $1<\beta \leq 2$. A by-product of the author's work is the development of a general framework for analyzing dynamical systems of the type introduced by Connes and Marocolli. These tools involve the use of Hecke operators, ergodic theory and  equidistribution of Hecke points \cite{clozel2001hecke}.

In this paper we combine the approach of Laca, Larsen and Neshveyev \cite{laca2007phase} together with the work of Ha and Paugan \cite{ha2005bost} and generalize the work of Connes and Marcolli to the explicit case of the Shimura variety $(GSp_{4},\mathbb H_{2}^{\pm})$. There are several complications compared to the $GL_{\mathbb Q,2}$-system. First, the action of the group $Sp_{4}(\mathbb Z)$ on the set $\mathbb H_{2}^+ \times GSp_{4}(\mathbb A_{\mathbb Q,f})$ is not free and one can not easily resolve this issue by excluding the subset $\mathbb H_{2}^{+}\times \{0_4\}$. Our approach consists of first replacing the homogeneous space $\mathbb H_{2}^{+}$ by the quotient $K\backslash PGSp_4^+(\mathbb R)$ (where $K$ is a compact subgroup of $PGSp_{4}(\mathbb R)^+$) and then prove a one-to-one correspondence of the $\textmd{KMS}_\beta$ states between the two systems. We first establish the correspondence between the set of $\textmd{KMS}_\beta$ states and Borel measures, which allows us to study the properties of those measures  instead of working directly with the $\textmd{KMS}_\beta$ states. The second difficulty arises from the structure of the Hecke pair $(GSp_{2n}(\mathbb Q),Sp_{2n}(\mathbb Z))$. As we show in this paper, the case $n=2$ is already computationally demanding and even in this case it is not always possible to directly apply some techniques used in \cite{laca2007phase} (especially in the critical interval $3 < \beta \leq 4$). As a first result we show in Theorem \ref{label 68} that the $GSp_{4}$-system does admit any $\textmd{KMS}_{\beta}$ state for $0 <\beta <3$ and $\beta\notin \{1,2\}$. We next show that the extremal states in the region $\beta > 4$ correspond to Gibbs states and give an explicit construction of these states in Theorem \ref{label 75}. The final main result (Theorem \ref{label 100}) is uniqueness theorem: we show that in the region $3<\beta \leq 4$, the $GSp_{4}$-system admits a unique $\textmd{KMS}_\beta$ state. To show this, we split the proof into two parts. The main ingredient of the first is the convergence of Dirichlet $L$-functions for non-trivial characters. The second part relies on a variant of the technique used in \cite{laca2007phase}. As stated above, the structure of the Hecke pair $(GSp_{2n}(\mathbb Q),Sp_{2n}(\mathbb Z))$ becomes less explicit for $n\geq 2$ and in order to compute the number of right representatives in a double coset one has to work with upper bounds instead of explicit formulas. We achieve this by using the root datum of the group $GSp_{2n}$ and use the equidistribution of Hecke points for the group $GSp_{2n}$ to establish the second ergodicity result.

\subsection*{Acknowledgment}
The author would like to thank his advisor Matilde Marcolli for her guidance throughout this project. The author is also grateful to George Elliott, Serguey Neshveyev and Jean Renault for their valuable comments on this work.

\par This work has been partially supported by the NSERC Postgraduate Scholarship PGSD \textmd{2-535022-2019}.
\subsection*{Notations and conventions}

\hfill\\
We use the common notations $\mathbb N,\mathbb Z,\mathbb Q,\mathbb R,\mathbb C$ together with  $\mathbb R^*_+=(0,+\infty)$ and  $\mathbb Z^+=R^*_+ \cap \mathbb Z$;

If $R$ is a ring, we denote its group of multiplicative units by $R^{\times}$;\\
We use the notation $\textmd{Mat}_{n}(R)$ for the ring of square matrices with entries in $R$. We denote by $E_{ij}$ the usual elementary matrix with $1$ in the $(i, j)$ position and $0$
elsewhere;\\
The group of units in the ring $\textmd{Mat}_{n}(R)$ is denoted by $GL_n(R)$. If $A$ is a square matrix, then $A^t$ stands for its transpose. If $A_1,\dots,A_n$ are square matrices we denote by $\textmd{diag}(A_1,\dots,A_n)$ the square matrix with $A_1,\dots, A_n$ as diagonal blocks and $0$'s otherwise. We use $1_n$ and $0_n$ to denote the $n\times n$ identity matrix and the a rectangular zero matrix;\\
$\abs{F}$ denotes the cardinality of a finite set $F$;\\
the set of prime numbers is denoted by $\mathcal{P}$;\\
given a nonempty finite set of prime numbers $F \subset \mathcal{P}$, we denote by $\mathbb N (F)$ the unital multiplicative subsemigroup of $\mathbb N$ generated by $p\in F$;\\
we write $f(x) = O(g(x))$ if $\abs{f(x)}/g(x)$ is bounded at $+\infty$;\\ 
for two sequences $\{a_n\}$ and $\{b_n\}$, we write $a_n \sim b_n$ if $\lim_{n} (a_n/b_n) =1$ and $$\sum_n a_n \sim \sum_n b_n$$ if the two series are simultaneously divergent or convergent;\\
if $Y$ is subset of $X$, we denote $Y^c = X \backslash Y =\{a\in X: a\notin Y\}$;\\
for a number field $K$, we denote by $\mathbb A_{K}= \mathbb A_{K,f}\times \mathbb A_{K,\infty}$ the ad\`ele ring of $K$, where $\mathbb A_{K,f}$ is the ring of finite ad\`eles and $\mathbb A_{K,\infty}$ the infinite ad\`eles of $K$. The ring of integers of $K$ is denoted by $\mathcal{O}_{K}$.

\section*{Background} 

\subsection{ Operator algebraic formulation of quantum statistical mechanics} 

We briefly review the operator algebraic formulation of quantum statistical mechanics. For a more comprehensive treatment of this material, we refer the reader to \cite{bratellirobinson1}, \cite{bratellirobinson2}.

\par Given an absract $C^*$-algebra $\mathcal{A}$, we know from Gelfand–Naimark theorem that $\mathcal{A}$ is $^*$-isomorphic to a $^*$-subalgebra of the algebra of bounded operators on a Hilbert space. This result, together with the axioms of quantum mechanics, motivates the following operator algebraic formulation of quantum statistical mechanics. 
\begin{definition}
A quantum statistical mechanical system $(\mathcal{A},(\sigma_t)_{t\in \mathbb R})$ is a $C^*$-algebra $\mathcal{A}$ together with a strongly continuous one-parameter group of automorphisms $(\sigma_{t})_{t\in \mathbb R}$; that is, the map $$t \rightarrow \sigma_t(a)$$ is norm continuous for every $a\in \mathcal{A}$.
\end{definition}

 We also say that the pair $(\mathcal{A},(\sigma_t)_{t\in \mathbb R})$ is a $C^*$-dynamical system. One should view $\mathcal{A}$ as the algebra of observables of a quantum system with time evolution implemented by the one-parameter group $(\sigma_t)_{t\in \mathbb R}$. If the algebra $\mathcal{A}$ is unital (with unit element $e$), a \textit{state} on $\mathcal{A}$ is a linear functional $\phi: \mathcal{A} \rightarrow \mathbb C$ satisfying the following normalization and positivity conditions: 
\begin{align}
    \phi(e)=1,\quad
    \phi(a^*a)\geq 0 \nonumber
\end{align} 
 
If $a \in \mathcal{A}$ is self adjoint, we have a decomposition of the form $a=a^+ -a^-$ where $a^+,a^- \in \mathcal A^+$ and one should think of $\phi(a)$ as the expectation value of the observable $a$ in the  physical state $\phi$.

\par When the algebra $\mathcal{A}$ is non-unital we shall always work with \textit{weights} (which we will define shortly) first and then replace the normalization condition $\phi(e)=1$ by
 
 \begin{equation}\label{label 69}
    \norm{\phi}:=\sup_{x\in \mathcal{A},\norm{x}\leq 1} |\phi(x)| =1
\end{equation}
 
to get a state. A weight on $\mathcal{A}$ is a function $\phi: \mathcal{A}^+ \rightarrow [0,\infty]$ (here $\mathcal{A}^+$ is the convex cone of positive elements in $\mathcal{A}$) such that $\phi(\lambda a +b )=\lambda \phi(a)$ and  $\phi(\lambda a )=\lambda \phi(a)$ for $\lambda\in \mathbb R^+$ and all $a,b\in \mathcal{A}^+$. If the algebra  $\mathcal{A} $ is unital, any weight can be written as $\phi= \lambda \omega$, where $\lambda>0$ and $\omega$ is a state. Hence the two notions of states and weights essentially coincide in the unital case. In general, this is not true in the nonunital case. As noted in \cite{christensen2020structure}, one should think of states as probability Borel measures on noncommutative spaces while weights correspond to regular Borel measures in the commutative case.

\par In statistical mechanics we are  interested in the so-called thermal equilibrium states at different temperatures. The Kubo-Martin-Shwinger (\textmd{KMS}) \cite{haag1967equilibrium, LocalQuantumPhycis} condition at inverse temperature $\beta$ was proposed in \textmd{1961} by Haag, Winnik and Hugenholtz as an equilibrium condition in the $C^*$-algebraic setting of statistical mechanics. We recall the notion of a $\textmd{KMS}_\beta$-weight from \cite{combes1971poids}. See and \cite{bratellirobinson1} \cite{bratellirobinson2} for a more detailed discussion on $\textmd{KMS}_\beta$-weights.

\begin{definition} \label{labe 77}
Let $A$ be a $C^*$-algebra, $\phi$ a weight on $A$ and $\sigma_t$ a strongly continuous one-parameter group of automorphisms of $A$. We set  $\mathcal{N}_{\phi} =\{a\in A  \mid \phi(a^*a)< \infty\}$ and let $\beta >0$. We say that $\phi$ is a $\textmd{KMS}_\beta$-weight if:

\begin{enumerate}
    \item $\phi \circ \sigma_t=\phi$ for every $t\in \mathbb R$
    \item For every $a,b \in \mathcal{N}_{\phi} \cap \mathcal{N}_{\phi}^* $, there exists a bounded continuous function $F$ on the closed strip $\Omega=\{x\in \mathbb C \mid 0 \leq  \textmd{im} \leq \beta \}$ and holomorphic on $\Omega^{0}$ such that
    
    \begin{equation*}
        F(t)=\phi(a\sigma_t(b)),\quad F(t+i\beta)=\phi(\sigma_t(b)a)
    \end{equation*}
\end{enumerate}
\end{definition}

Although this was the original definition of the $\textmd{KMS}_\beta$ condition, in practice the following equivalent characterization (See \cite[Theorem 6.36]{kustermans1997kms}) is often used.

\begin{prop}
Let $\phi$ be a weight on a $C^*$-dynamical system $(\mathcal{A},(\sigma_t)_{t\in \mathbb R})$. Then $\phi$ is a $\textmd{KMS}_\beta$ weight if and only if
\begin{enumerate}
    \item $\phi \circ \sigma_t=\phi$ for every $t\in \mathbb R$
    \item For every $\sigma$-analytic element $a$ in $\mathcal{A}$, we have
    
    \begin{equation*}
    \phi(aa^*)=\phi (\sigma_{\frac{i\beta}{2}} (a)^*  \sigma_{\frac{i\beta}{2}} (a)).
\end{equation*}
\end{enumerate}
\end{prop}

\subsection{Groupoid algebras and Hecke pairs}

An important class of $C^*$-dynamical systems arises as the algebra of compactly supported functions on (locally compact) topological groupoids. We review, without proofs, the general results related to these systems and derived in the first two sections of \cite{laca2007phase}. For general information about grouopoids and groupoid $C^*$-algebras, we refer the reader to \cite{renault2006groupoid}.

\par Consider a countable group $G$ acting on a locally compact second countable topological space $X$. The transformation groupoid is the space $G\times X$ with unit space $X$ with the source and target maps given by $s(g,x)=x$ and $t(g,x)=gx$

and the composition is defined by

\begin{equation*}
    (g,x)(h,y)=(gh,y) \quad \text{if}\,\,\, x=hy.
\end{equation*}

If $\Gamma$ is a subgroup of $G$ and the action of $\Gamma$ is free and proper, we introduce a new groupoid $\Gamma \backslash G \times_{\Gamma} X$ by taking the quotient of $G\times X$ by the following action of $\Gamma \times \Gamma $:

\begin{equation}
    (\gamma_1,\gamma_2)(g,x):=(\gamma_1 g \gamma_2^{-1},\gamma_2x). \label{label 63}
\end{equation}

In all our settings, the main motivation for taking the quotient by this action is physical. In fact, to obtain a well behaved partition function of the $C^*$-dynamical system we will introduce shortly, one should necessarily take the quotient by the group $\Gamma \times \Gamma$ (See also \cite{connes2004physics} for another motivation based on the $\mathbb Q$-lattice picture).\\
The groupoid $\mathcal{G}=\Gamma \backslash G\times_{\Gamma} X$ is an \textit{\'etale groupoid}, i.e. the source map $s$ (and hence the target map) is a local homeomorphism. In particular, it has discrete fibers \cite[Prop. 2.8]{renault2006groupoid}. In this case, we introduce the algebra $C_c(\mathcal{G})$ of continuous compactly supported functions on the quotient space $\Gamma \backslash G \times_{\Gamma} X$ and we define the convolution of two such functions by

\begin{equation}
        (f_1* f_2)(\omega)= \sum_{\omega_1\omega_2=\omega} f_1(\omega_1)f_2(\omega_2)=\sum_{\omega_1\in \mathcal{G}^{t(\omega)}} f_1(\omega_1)f_2(\omega_1^{-1}\omega). \label{label 61}
\end{equation}
Notice that this is a finite sum since $f_1$ has compact support and the fibers are discrete. If the action of $\Gamma$ is proper but not free, the quotient space $\Gamma \backslash G \times_{\Gamma} X$ is no longer a groupoid (cf. Proposition \ref{label 74}). Under the assumption that $X$ is a homogeneous space of the form $\tilde{X}/H$, where now the action of $G$ on $\tilde{X}$ is free and proper, we can still define a natural convolution algebra from the groupoid algebra $C_c (\Gamma \backslash G \times_{\Gamma} \tilde{X})$ as was done in the $GL_2$-case in \cite{connes2004physics}. More specifically, viewing the elements of $C_c (\Gamma \backslash G \times_{\Gamma} \tilde{X})$ as $\Gamma \times \Gamma$-invariant functions, we rewrite the convolution product \eqref{label 61} as

\begin{equation}
    (f_1* f_2)(g,x)= \sum_{s\in \Gamma \backslash G} f_1(gs^{-1},sx)f_2(s,x). \label{label 62}
\end{equation}

We then define a convolution algebra on the quotient $\Gamma \backslash G \times_{\Gamma} X$ by restricting the convolution product \eqref{label 62} to weight zero functions on $C_c (\Gamma \backslash G \times_{\Gamma} \tilde{X})$, namely functions satisfying 

\begin{equation*}
    f(g, x \alpha)= f(g,x),\quad \forall \alpha \in H.
\end{equation*}

Define an involution on $C_c(\Gamma \backslash G\times_{\Gamma} X)$ by 

\begin{equation*}
    f^*(\omega)=\overline{f(\omega^{-1})}.
\end{equation*}

For each $x\in X$, we have a $*$-representation $\pi_x$ of $C_c(\Gamma \backslash G\times X)$ on the Hilbert space $l^{2}(\Gamma \backslash G)$ defined by

\begin{equation*}
    (\pi_x)(f) \delta_{\Gamma h} = \sum_{g\in \Gamma \backslash G} f(gh^{-1},hx) \delta_{\Gamma g},\quad f\in C_c(\Gamma \backslash G\times_{\Gamma} X).
\end{equation*}

One can show that the operators $\pi_x(f)$ are uniformly bounded \cite{ha2005bost,laca2007phase} and we denote by $\mathcal{B}=C_r^*(\Gamma \backslash G\times X)$ the completion of $C_c(\Gamma \backslash G\times X)$ in the reduced norm 

\begin{equation}
    \norm{f}=\sup_{x\in X}\norm{\pi_{x}(f)}. \label{label 64}
\end{equation}

In fact, it is easy to verify that that \eqref{label 64} defines a $C^*$-seminorm. The fact that we get a norm follows from the identity
\begin{equation*}
    \langle \pi_{x}f \delta_{\Gamma g_1},\delta_{\Gamma g_2} \rangle = f(g_2 g_1^{-1},g_1x).
\end{equation*}

Let $Y$ be any clopen $\Gamma$-invariant subset of $X$ and denote by $\Gamma \backslash G \boxtimes_\Gamma Y$ the quotient of the space 

\begin{equation*}
    \{(g,y) \mid g\in G,\,\,\, y\in Y,\,\,\, gy\in Y\},
\end{equation*}

by the action of $\Gamma \times \Gamma$ defined in \eqref{label 63}. We denote by $C_c(\Gamma \backslash G \boxtimes_\Gamma Y)$ the algebra of compactly supported functions on $\Gamma \backslash G \boxtimes_\Gamma Y$ with the convolution product given by

\begin{equation*}
    (f_1* f_2)(g,y)= \sum_{\substack{s\in \Gamma \backslash G\\
    sy \in Y}} f_1(gs^{-1},sy) f_2(s,y),
\end{equation*}

and involution

\begin{equation*}
    f^*(g,y)=\overline{f(g^{-1},gy)}.
\end{equation*}

We let $\mathcal{A}=C_r^* (\Gamma \backslash G \boxtimes_\Gamma Y)$ be the corner algebra $e\mathcal{B}e$, where $e$ is the $\Gamma \times \Gamma$-invariant function on $ G\times X$ defined by

\begin{equation*}
    e  (g,x)= 
\begin{cases}
1 \quad \text{if}\,\,\, (g,x) \in \Gamma \times Y\\
0\quad \text{otherwise}.
\end{cases}
\end{equation*}

Given $x\in X$, we put 

\begin{equation}\label{label 72}
    G_x=\{g\in G \mid gx \in Y\}.
\end{equation}

Then we have a representation of $C_c(\Gamma \backslash G\boxtimes_\Gamma Y)$ on the Hilbert space $\mathcal{H}_x=l^2(\Gamma \backslash G_x)$ given by

\begin{equation*}
    \pi_{x}(f) \delta_{\Gamma h}= \sum_{g\in \Gamma \backslash G_x} f(gh^{-1},hx) \delta_{\Gamma g},\quad f\in C_c(\Gamma \backslash G\boxtimes_\Gamma Y),
\end{equation*}

and the algebra $\mathcal{A}$ coincides (\cite{laca2007phase}) with the completion of $C_c(\Gamma \backslash G\boxtimes_\Gamma Y)$ in the norm defined by

\begin{equation*}
    \norm{f}=\sup_{y\in Y}\norm{\pi_{y}(f)}.
\end{equation*}

Assume that we are given a homomorphism 
\begin{equation*}
    N: G \longrightarrow \mathbb R _+^{*},
\end{equation*}

such that $\Gamma \subseteq \ker(N)$. We then define a one-parameter group of automorphisms of $\mathcal{B}$ by

\begin{equation*}
    \sigma_t(f)(g,x)= N(g)^{it} f(g,x),\quad \text{for}\,\,\, f\in C_c(\Gamma \backslash G\times X).
\end{equation*}

The operator on $l^2(\Gamma \backslash G)$ given by

\begin{equation*}
    H_x \delta_{\Gamma g} = \log N(g) \cdot \delta_{\Gamma g}
\end{equation*}

is the Hamiltonian and the dynamics $\sigma_t$ is then spatially implemented as

\begin{equation*}
    \pi_{x}(\sigma_t(a))= e^{itH_x} \pi_x(a) e^{-it H_x},\quad \forall x\in X, \forall a\in \mathcal B.
\end{equation*}

The following result will be the starting point of our $\textmd{KMS}_\beta$-analysis of the dynamical system $(\mathcal{A},\sigma_t)$. 

\begin{prop}\label{label 65}
Let $G$, $X$ and $Y$ as described earlier and suppose $\Gamma$ acts freely on $X$. Then for $\beta >0$ there exists a one-to-one correspondence between $\textmd{KMS}_{\beta}$ weights $\phi$ on $\mathcal{A}$ with domain of definition containing $C_c(\Gamma \backslash Y)$ and Radon measures $\mu$ on $Y$ such that

\begin{equation*}
    \mu(gB)=N(g)^{-\beta}\mu(B)
\end{equation*}
\end{prop}

for every $g\in G$ and every Borel compact subset $B\subseteq Y$ such that $gB \subseteq Y$. If $\nu$ denotes the induced measure on $\Gamma \backslash Y$, then the corresponding weight $\phi$ is given by

\begin{equation}\label{label 70}
    \phi(f)=\int_{\Gamma \backslash Y} f(e,y) d\nu (y),\quad f\in \mathcal{A}.
\end{equation}
 
\begin{proof}
See \cite[Proposition 2.1]{laca2007phase}
\end{proof}

Recall that if $G$ be a group and $\Gamma$ a subgroup, the pair $(G,\Gamma)$ is called is called a \textit{Hecke pair} if for any $a \in G$
\begin{equation*}
    [\Gamma : \Gamma \cap a^{-1} \Gamma a] < \infty.
\end{equation*}

If $(G,\Gamma)$ is a Hecke pair then every double coset of $\Gamma$ contains finitely many right and left cosets of $\Gamma$:
\begin{equation*}
    \Gamma a \Gamma = \bigsqcup_ {\gamma \in \Gamma \backslash (\Gamma \cap a \Gamma a^{-1}) } \gamma a\Gamma = \bigsqcup_ {\gamma \in \Gamma \backslash (\Gamma \cap a^{-1} \Gamma a) } \Gamma g \gamma,
\end{equation*}
so $\abs{\Gamma \backslash \Gamma a \Gamma} =  [\Gamma : \Gamma \cap a^{-1} \Gamma a]$. We denote the cardinality of this set by $\deg_{\Gamma}(a)$.

Let $\beta \in \mathbb R$ and $S$ is a semisubgroup of $G$ containing $\Gamma$. Then we define 

\begin{equation}
    \zeta_{S,\Gamma}(\beta):= \sum_{s\in \Gamma \backslash S} N(s)^{-\beta}=\sum_{s\in 
    \Gamma \backslash S/\Gamma } N(s)^{-\beta} \deg_{\Gamma}(s).
\end{equation} \label{label 20}

Let $G$ be a group acting on a set $X$ and suppose $(G,\Gamma)$ is a Hecke pair. The \textit{Hecke operator} associated to $g\in G$ is the operator $T_g$ on $\Gamma$-invariant functions on $X$ defined by

\begin{equation}
    (T_gf) (x) = \frac{1}{\deg_\Gamma(g)} \sum_{h\in \Gamma \backslash \Gamma g \Gamma } f(hx). \label{label 24}
\end{equation}

\subsection {Abstract Bost-Connes-Marcolli systems}
\label{Abstract BCM}
In this section, we briefly recall, without proofs, the general properties of abstract Bost-Connes-Marcolli systems introduced in \cite{ha2005bost}.\\
A BCM datum is a tuple $\mathcal{D}=(G,X,V,M)$ with $(G,X)$ a Shimura datum, $(V,\psi)$ a faithful representations of $G$ and $M$ an envelopping semigroup (see definition \ref{label 82}) for $G$ contained in $ \textrm{End}(V)$. A level structure on $\mathcal{D}$ is a triple $\mathcal{L}=(L,K,K_M)$ with $L\subseteq V$ a lattice, $K\subseteq G(\mathbb{A}_f)$ a compact subgroup and $K_M\subseteq M(\mathbb{A}_f)$ a compact open subsemigroup such that
\begin{itemize}
    \item $K_M$ stabilizes $L \otimes_{\mathbb{Z}} \hZ$
    \item $\psi(K)$ is contained in $K_M$.
\end{itemize}
The pair $(\mathcal{D},\mathcal{L})$ is called a BCM pair. We let 
\begin{equation*}
    Y_{\mathcal{D},\mathcal{L}}=K_M \times \textrm{Sh}(G,X),
\end{equation*}
and we denote the points of $ Y_{\mathcal{D},\mathcal{L}}$ by $y=(\rho,[z,l])$. We let $Y^\times_{\mathcal{D},\mathcal{L}}$ be the set of invertible elements $y=(\rho,[z,l])$ in  $Y_{\mathcal{D},\mathcal{L}}$. We have a partially defined action of $G(\mathbb{A}_f)$ on $ Y_{\mathcal{D},\mathcal{L}}$:
\begin{equation*}
    g\cdot y= (gy,[z,lg^{-1}])\quad \text{for}\,\,\, y =(\rho,[z,l]).
\end{equation*}
We consider the subspace 

\begin{equation*}
    \mathcal{U}_{\mathcal{D},\mathcal{L}} \subseteq G(\mathbb{A}_f)\times Y_{\mathcal{D},\mathcal{L}}
\end{equation*}
of pairs $(g,y)$ such that $g \cdot y \in Y_{\mathcal{D},\mathcal{L}}$. This space is a groupoid with source and target maps $s:\mathcal{U}_{\mathcal{D},\mathcal{L}} \rightarrow Y_{\mathcal{D},\mathcal{L}}$ and $t:\mathcal{U}_{\mathcal{D},\mathcal{L}} \rightarrow Y_{\mathcal{D},\mathcal{L}}$ given by $s(g,y)=y$ and $t(g,y)=gy$. The unit space is $Y_{\mathcal{D},\mathcal{L}}$ and composition is defined as 
$$(g_1,y_1)\circ (g_2,y_2)=(g_1g_2,y_2) \quad \text{if}\,\,\, y_1=g_2y_2.$$

There is an action of $K^2$ on the groupoid $\mathcal{U}_{\mathcal{D},\mathcal{L}}$ given by

\begin{equation*}
    (\gamma_1,\gamma_2) \cdot (g,y) := (\gamma_1g\gamma_2^{-1},\gamma_2 y)
\end{equation*}

and the quotient stack $\mathfrak{Z}_{\mathcal{D},\mathcal{L}}=[K^2 \backslash \mathcal{U}_{\mathcal{D},\mathcal{L}}]$ has the structure of a stack-groupoid (see \cite[Appendix A]{ha2005bost} ).\\
From now on we suppose that that $(\mathcal{D},\mathcal{L})$ is a BCM pair such that the Shimura datum $(G,X)$ is classical, i.e

\begin{equation*}
    \textrm{Sh}(G,X)=G(\mathbb{Q}) \backslash  G(\mathbb{\mathbb{A}}_f) \times X.
\end{equation*}
We let $\Gamma= G(\mathbb{Q}) \cap K$ and 
\begin{equation*}
    \mathcal{U}^{\textmd{princ}} :=\{ (g,\rho,z)\in G(\mathbb{Q}) \times K_M \times X \mid g\rho \in K_M)\}.
\end{equation*}
We let $X^{+}$ be a connected component of $X$, $G(\mathbb{Q})^+=G(\mathbb{Q}) \cap G(\mathbb{R})^+$ (where $G(\mathbb{R})^+$ is the identity component of $G(\mathbb{R})$) and $\Gamma_+=G(\mathbb{Q})^+ \cap K$. Consider the groupoid
\begin{equation*}
    \mathcal{U}^+=\{ (g,\rho,z)\in G(\mathbb{Q})^+ \times K_M \times X^+ \mid g\rho \in K_M)\},
\end{equation*}
with the composition given by

\begin{equation}
    (g_1,\rho_1,z_1) \circ (g_2,\rho_2,z_2) = (g_1g_2, \rho_2,z_2) \quad \textmd{if}\, (\rho_1,z_1)=(g_2\rho_2,\rho_2z_2) \label{label 5}
\end{equation}

There is a natural action of $\Gamma^2$ (resp. $\Gamma_+^2$) on $\mathcal{U}^{\textmd{princ}} $ (resp. $ \mathcal{U}^+$) given by

\begin{equation}
(\gamma_1,\gamma_2)\cdot (g,\rho,z):= (\gamma_1g \gamma_2^{-1},\gamma_2\rho,\gamma_2 z)    \label{label 6}
\end{equation}
and the quotient stack $\mathfrak{Z}_{\mathcal{D},\mathcal{L}}^{\textmd{princ}}$ (resp. $\mathfrak{Z}_{\mathcal{D},\mathcal{L}}^{+}$) of $\mathcal{U}^{\textmd{princ}}$ (resp. $\mathcal{U}^{+}$) by $\Gamma^2$ (resp. $\Gamma_+^2$) has again the structure of a stack-groupoid.

For an arbitrary BCM pair $(\mathcal{D},\mathcal{L})$, the relation between the three groupoids  $\mathfrak{Z}_{\mathcal{D},\mathcal{L}}$ and $\mathfrak{Z}_{\mathcal{D},\mathcal{L}}^{princ}$, $\mathfrak{Z}_{\mathcal{D},\mathcal{L}}^{+}$ is unclear a priori. The following important result obtained in \cite[Propositions 5.2, Proposition 5.3]{ha2005connes} provides sufficient conditions for these three stack-groupoids to coincide.

\begin{prop} \label{label 60}
We denote by $h(G,K)$ the cardinality of the finite set $G(\mathbb{Q}) \backslash G(\mathbb{A}_f)/K$. Assume that $h(G,K)=1$ and the natural map $\Gamma \rightarrow G(\mathbb Q) / G(\mathbb Q)^{+}$ is surjective. Then the natural maps

\begin{equation*}
\mathfrak{Z}_{\mathcal{D},\mathcal{L}}^{+} \longrightarrow \mathfrak{Z}_{\mathcal{D},\mathcal{L}}^{princ}, \quad \mathfrak{Z}_{\mathcal{D},\mathcal{L}}^{princ} \longrightarrow \mathfrak{Z}_{\mathcal{D},\mathcal{L}}     
\end{equation*}

are isomorphisms.
\end{prop}

Hence when the condition of Proposition \ref{label 60} are satisfied, it is enough to work with the BCM system associated to one of the three stack-groupoids. Let 

\begin{equation*}
    \mathcal{H}(\mathcal{D},\mathcal{L}) = C_c(Z_{\mathcal{D},\mathcal{L}})
\end{equation*}

be the algebra of continuous compactly supported functions on the coarse quotient $Z_{\mathcal{D},\mathcal{L}}$ of $\mathcal{U}$ by the action of $K^2$. We view its elements as functions on $\mathcal{U}_{\mathcal{D}, \mathcal{L}}$ satisfying the following properties:

\begin{equation*}
    f(\gamma g,y)=f(g,y), \quad f(g\gamma,y)=f(g,\gamma y),\quad \forall \gamma \in K,\,\,\, g\in G(\mathbb A _ f),\,\,\, y \in Y_{\mathcal{D}, \mathcal{L}}
\end{equation*}

The convolution product on $\mathcal{H}(\mathcal{D},\mathcal{L})$ is defined by the expression

\begin{equation}
    (f_1* f_2)(g,y):=\sum_{\substack{h\in K \backslash G(\mathbb A_f)\\ hy \in Y_{\mathcal{D,L}}}} f_1(gh^{-1},hy)f_2(h,y),
\end{equation}

and the involution is given by 

\begin{equation*}
    f^*(g,y):=\overline{f(g^{-1},gy)}.
\end{equation*}

Let $y=(\rho, [z,l]) \in Y_{\mathcal{D},\mathcal{L}}$ and we put $G_y= \{g\in G(\mathbb A_f) \mid g\rho \in K_M\}$. We define a $^*$-representation 
$\pi_y: \mathcal{H}(\mathcal{D}, \mathcal{L}) \rightarrow \mathcal{B}( l^2(K \backslash G_y ))$ by 
\begin{equation*}
    (\pi_y(f)\xi)(g):= \sum_{h\in K \backslash G_y } f(gh^{-1},hy) \xi(h), \quad f\in \mathcal{H}(\mathcal{D},\mathcal{L})
\end{equation*}
where $\xi$ is the standard basis of $l^2(K \backslash G_y )$. The operators $\pi_y(f)$, for $y\in Y_{\mathcal{D},\mathcal{L}}$ and $f \in  \mathcal{H}(\mathcal{D},\mathcal{L}) $ are uniformly bounded \cite[Lemma 4.16]{ha2005bost} and we obtain a $C^*$-algebra $\mathcal{A}$ after completing  $\mathcal{H}(\mathcal{D},\mathcal{L})$ in the norm
\begin{equation*}
    \norm{f}= \sup_{y\in Y_{\mathcal{D},\mathcal{L}}} \norm{\pi_y(f)}.
\end{equation*}

Given a homomorphism 

\begin{equation*}
    N: \textmd{GL}(V) \rightarrow \mathbb R_+^*,
\end{equation*}

we define a time evolution on $\mathcal{H}(\mathcal{D},\mathcal{L})$ by
\begin{equation*}
    \sigma_t(f)(g,y)= N(\psi(g))^{it} f(g,y),
\end{equation*}

so that the operator on $l^2(K \backslash G_y)$ given by 

\begin{equation*}
    (H_y\zeta)(g)=\log N(\psi(g)) \zeta (g)
\end{equation*}
is the Hamiltonian. The resulting $C^*$-dynamical system $(\mathcal{A},\sigma_t)$ is the Bost-Connes-Marcolli system associated to the $\textmd{BCM}$ pair $(\mathcal{D},\mathcal{L})$.\\

The zeta function associated to the \textmd{BCM} pair $(\mathcal{D},\mathcal{L})$ plays an important role in the $\textmd{KMS}_\beta$ analysis of the system $(\mathcal{A},\sigma_t)$. It is defined as the complex-valued series

\begin{equation} \label{label 84}
    \zeta_{\mathcal{D},\mathcal{L}}(\beta):= \sum_{g\in \textmd{Sym}_f^{\times} \backslash \textmd{Sym}_f  } N(\psi(g))^{-\beta},
\end{equation}

where  \textmd{Sym}$_f(\mathcal{D},\mathcal{L}):= \psi^{-1}(K_M)$ and $\textmd{Sym}_f^{\times}(\mathcal{D},\mathcal{L})$ denotes the group of invertible elements in $\textmd{Sym}_f(\mathcal{D},\mathcal{L})$. The pair $(\mathcal{D},\mathcal{L})$ is called summable of there exists $\beta_{0} \in \mathbb R$ such that  $\zeta_{\mathcal{D},\mathcal{L}}(\beta)$ converges in the right plane  $\{\beta \in \mathbb C \mid \textmd{Re}(\beta) > \beta_0\}$ and extends to a meromorphic function on the full complex plane.  By \cite[Proposition 4.19]{ha2005bost} we know that if $y\in Y^\times_{\mathcal{D},\mathcal{L}}$ , then $G_y=\textmd{Sym}_f(\mathcal{D},\mathcal{L})$ and the zeta function $\zeta_{\mathcal{D},\mathcal{L}}(\beta)$ coincides with the partition function 

$$\zeta_{y}(\beta)= \frac{1}{|K \backslash K_0|} Tr (e^{-\beta H_y}), \quad K_0= \psi^{-1}(K_M^{\times}),$$

of the system $(\mathcal{A},\sigma_t)$.

\section{Bost-Connes-Marcolli system for the Siegel Modular Variety}

\subsection{The Symplectic Group}

Let $n\in \mathbb{N}$ and $R$ be a commutative unital ring. The \textit{symplectic group of similtitudes of degree $n$}  is defined by
\begin{equation*}
    GSp_{2n} (R)=\{g\in GL_{2n}(R): \exists\, \lambda(g)\in R^{\times}\,| \, g^t\,\Omega\,g\,=\,\lambda(g)\,\Omega \},
\end{equation*}
where 
\begin{equation*}
    \Omega= \mqty (0& 1_n \\ 
    -1_n & 0 ),\quad \text{ $1_n$ is the $n\times n$ identity matrix}.
\end{equation*}
The function $\lambda: GSp_{2n}(R) \rightarrow R^\times$ is called the multiplier homomorphism. Its kernel is the symplectic group $Sp_{2n}(R)$ and there is an exact sequence
\begin{equation}
    1 \longrightarrow Sp_{2n}(R) \longrightarrow GSp_{2n} (R) \longrightarrow R^{\times} \longrightarrow 1. \label{label1}
\end{equation}
If $g=\mqty(A & B \\ C & D)\in GSp_{2n}(R)$ then the following assertions are equivalent:
\begin{enumerate} [label=(\roman*)]
    \item $\lambda(g)=\lambda(g^t)$
    \item The inverse of the matrix $g$ is given by:
    
    \begin{equation}
    g^{-1}=\lambda(g)^{-1}\mqty(D^t & -B^t \\ -C^t & A^t) \label{label 17}
    \end{equation}
    
    \item The blocks $A,B,C,D$ satisfy the conditions
    \begin{equation}
        A^tC=C^tA,\, B^tD=D^tB,\, A^tD-C^tB=\lambda(g) 1_n \label{label 7}
    \end{equation}
     \item The blocks $A,B,C,D$ satisfy the conditions
    \begin{equation}
        A^tB=B^tA,\, C^tD=D^tC,\, A^tD-B^tC=\lambda(g) 1_n \label{label 8}
    \end{equation}
\end{enumerate}

For $r\in R^\times$, we put

\begin{equation*}
    S_n(r):= \{g\in GSp_{2n}(R)\mid \lambda(g)=r\}.
\end{equation*}

We then obtain an embedding of symplectic groups of different degrees as follows. Given $r\in \mathbb{N}$ and $0 < j < n$, define the map

\begin{align*}
    S_n(q) \times S_{n-j}(q) \rightarrow S_n(q)\\
    (M_1,M_2) \mapsto M_1 \odot M_2,\\
\end{align*}

where

\begin{equation*}
        M_1 \odot M_2 := \mqty(A_1 & 0_{j \times (n-j)} & B_1 & 0_{j \times (n-j)} \\ 
                           0_{(n-j)\times j} & A_2 & 0_{(n-j)\times j} & B_2 \\
                           C_1 & 0_{j \times (n-j)} & D_1 & 0_{j \times (n-j)} \\ 
                           0_{(n-j)\times j} & C_2 & 0_{(n-j)\times j} & D_2),\quad M_1 =\mqty(A_1 & B_1 \\
                           C_1 & D_1),\quad M_2=\mqty(A_2 & B_2 \\ C_2 & D_2).
\end{equation*}
Note that

\begin{equation}
    (M_1 \odot M_2) \cdot (N_1 \odot N_2) = (M_1N_1) \odot (M_2N_2)
\end{equation}

Consider the following elements of $Sp_{2n}(R)$:
\begin{align}
    & \mqty(1_n & \alpha_1 E_{ii} \\ 0_n & 1_n),  
    \mqty(1_n & 0_n  \\ \alpha_2 E_{ii}& 1_n),
    \mqty(1_n & \alpha_3 (E_{ij}+E_{ji})  \\0_n & 1_n),
    \mqty(1_n & 0_n  \\ \alpha_4 (E_{ij}+E_{ji})  & 1_n), \nonumber \\
    & \mqty(1_n +\alpha_5 E_{ij} & 0_n \\ 0_n & 1_n -\alpha_5 E_{ji}), \label{label 27}
\end{align}
where $\alpha_1,\dots,\alpha_5 \in R$. If $F$ is a field, then the group $Sp_{2n}(F)$ is generated \cite{o1978symplectic} by the matrices given in \eqref{label 27}
with $\alpha_1,\dots,\alpha_5 \in F$.

As a connected reductive algebraic group, the center $Z$ of $G=GSp_{2n}$ consists of scalar matrices and the standard maximal torus is
\begin{equation*}
    T= \{\text{diag}(u_1,\dots,u_n,v_1,\dots,v_n)\,:\, u_1v_1=\dots=u_nv_n\neq 0\}.
\end{equation*}

If $t\in T$, we often write 
\begin{equation} \label{label 99}
    t=\textrm{diag}(u_1,\dots,u_n,u_1^{-1}\lambda(t),\dots,u_n^{-1}\lambda(t)),
\end{equation}

We  fix the following characters $e_i\in \textrm{Hom}(T,G_m)$:
\begin{equation*}
    e_i(t)=u_i,\quad i=0,1,\dots,n \,\,\, \textmd{where} \quad u_0:= \lambda(t).
\end{equation*}

and cocharacters $f_i \in \textrm{Hom}(G_m,T)$:
\begin{align*}
    f_0(u)& = \textrm{diag} (\underbrace{1,\dots,1}_n,\underbrace{u,\dots,u}_n),\\
    f_1(u)&= \textrm (\underbrace{u,1,\dots,1}_n,\underbrace{u^{-1},1,\dots,1}_n),\\
    \vdots\\
    f_n(u)&= \textrm (\underbrace{1,\dots,1,u}_n,\underbrace{1,\dots,1,u^{-1}}_n).
\end{align*}

\begin{prop} \label{label 37}

The root datum of $GSp_{2n}$ is described as follows. We set

\begin{align*}
    X &=\mathbb{Z}e_0 \oplus \mathbb{Z}e_1 \oplus \dots \mathbb{Z}e_n.\\
    X^\vee &= \mathbb{Z}f_0 \oplus \mathbb{Z}f_1 \oplus \dots \mathbb{Z}f_n.
\end{align*}

and let $\langle\cdot,\cdot\rangle$ the natural pairing on $X \times X^\vee$:
\begin{equation*}
    \langle e_i,f_j \rangle=\delta_{ij}.
\end{equation*}

Then we have the following set of simple roots:
\begin{equation*}
    \alpha_1(t)= u_{n-1}^{-1}u_n, \quad \dots \quad \alpha_{n-1}(t)=u_1^{-1}u_2,\quad \alpha_n(t)=u_1^2u_0^{-1},
\end{equation*}

where $t$ has the form in \eqref{label 99}. In terms of the basis $e_i,$ $i=0,1,\dots,n$, we have
\begin{equation*}
    \alpha_1=e_n-e_{n-1},\quad \dots \quad \alpha_{n-1}=e_2-e_1, \quad \alpha_n=2e_1-e_0.
\end{equation*}

The corresponding coroots are 
\begin{equation*}
    \alpha_1^\vee=f_n-f_{n-1},\quad \dots \quad \alpha_{n-1}^\vee=f_2-f_1, \quad \alpha_n^\vee=f_1.
\end{equation*}

Let $R=\{\alpha_1,\dots,
\alpha_n\}$ and $R^\vee=\{\alpha_1^\vee,\dots,
\alpha_n^\vee\}$. Then
\begin{equation*}
    (X,R,X^\vee,R^\vee)
\end{equation*}
is the root datum of $GSp_{2n}$. The Cartan matrix is given by
\begin{equation*}
  \langle \alpha_i,\alpha_j^{\vee} \rangle= \mqty (2 & -1 & & & & & \\ -1 & 2 & -1 & & & & \\ & -1 & 2 & -1 & & & \\ &  &  \ddots &  \ddots & \ddots \\ & & & -1 & 2 & -1 \\ & & & & -1 & 2 & -1 \\ & & & &  & -2 & 2).  
\end{equation*}
\end{prop}

\begin{proof}
See \cite[page 134-136]{tadic1994representations}
\end{proof}

\subsection{The Symplectic envelopping semigroup}

As noted in section \ref{Abstract BCM}, the abstract definition of the Bost-Connes-Marcolli system associated to a general Shimura datum $(G,X)$ requires the notion of an enveloping semigroup which plays the role of $\textmd{Mat}_{2,\mathbb Q}$ in the the case of the $GL_{2,\mathbb Q}$-system.

\begin{definition} \label{label 82}
Let $G$ be a reductive group over a field $F$. An enveloping semigroup for $G$ is a multiplicative semigroup $M$ which is irreducible and normal and such that $M^{\times}=G$.
\end{definition}

It is always possible to construct enveloping semigroup (see \cite[Appendix B.2]{ha2005bost}). For the case $G=GSp_{2n}$ we are considering in the paper, we have the following explicit description of $M$. Given a commutative $\mathbb Q$-algebra $R$ we have

\begin{equation}
    M(R):= MSp_{2n}(R)= \{m\in \textmd{Mat}_{2n}(R) \mid \exists\, \lambda(m)\in R, m^t \Omega m = \lambda(m)\Omega\},
\end{equation}

since $m \in MSp_{2n}(R)^{\times}$ if and only if $\lambda(m)\in R^{\times}$.

\subsection{The Siegel modular group}

The group $\Gamma_n=Sp_{2n}(\mathbb{Z})$ is called \textit{the Siegel modular group of degree $n$}. For $m\in \mathbb N$, we denote by $GL_{m}(\mathbb{Z})$ the unimodular group of degree $m$ and note that $\Gamma_n \subseteq  U_{2n}$ with equality if $n=1$. We let $MSp_{2n}^+(\mathbb Z)= \{M\in MSp_{2n}(\mathbb Z) \mid \lambda(M)>0\}$. Then  given any $M\in MSp_{2n}^+(\mathbb{Z})$, we denote by $\Gamma_n M \Gamma_n $ the double coset generated by $M$ and put
\begin{equation*}
    \mathcal{D} (\Gamma_n M \Gamma_n):= \{D\in \textrm{Mat}_{n}(\mathbb{Z})\, \textmd{such that there exists}\, \mqty (A & B \\ 0 & D)\in \Gamma_n M \Gamma_n \}.
\end{equation*}

\noindent For each $D\in \mathcal{D} (\Gamma M \Gamma)$, we set
\begin{equation*}
    \mathcal{B}(D,\Gamma_n M \Gamma_n):= \{B\in \textrm{Mat}_n(\mathbb{Z})\, \textmd{such that there exists} \mqty (A & B \\ 0 & D)\in \Gamma_n M \Gamma_n \}.
\end{equation*}
We define the following equivalence relation on $\mathcal{B}(D)$:
\begin{equation}
    B \sim B' \Leftrightarrow (B-B')D^{-1} \in \textrm{Sym}_{n}(\mathbb{Z}), \label{label 19}
\end{equation}
and we write $B \equiv B' \mod D$ if $B \sim B'$. For $r\in \mathbb{N}$, we put
\begin{equation*}
    S_n(r):= \{g\in GSp_{2n}(\mathbb{Z}) \mid \lambda(g)=r\}.
\end{equation*}

Let $M \in \textrm{Mat}_m(\mathbb{Z})$ with  $\det(M)> 0$ and $N\in S_{n}(r)$. Then by the Elementary Divisor Theorems (See \cite[Theorem 2.2, Chapter V]{krieg1990hecke}) the double cosets $U_m M U_m$ and $\Gamma_n N \Gamma_n$ contain  unique representatives of the form
\begin{align}
   \textmd{Elm}(M)= \textrm{diag} (a_1,a_2,\dots,a_m), \quad a_1,a_2,\dots,a_n \in \mathbb{N}, \label{label 39}\\
\end{align}
with $a_1 \mid a_2 \mid \dots \mid a_m$ and 

\begin{equation}
\textmd{Elm}(N)=\textrm{diag}(a_1,\dots,a_n,d_1,\dots,d_n), \quad a_1,\dots,a_n,d_1,\dots,d_n \in \mathbb{N},
\end{equation}

such that $a_id_i=r$, $i=1,\dots,n$ and $a_1\mid a_2\mid \dots a_n \mid d_n \mid \dots \mid d_{n-1} \mid d_1$.

\begin{thm}
Let $M\in MSp_{2n}^+(\mathbb{Z})$. Then a set of representatives of the right cosets relative to $\Gamma_n$ in $\Gamma_n M \Gamma_n$ is given by the matrices
$$\mqty(A & B \\ 0 & D), \quad A=\lambda(M) (D^t)^{-1} $$
where
\begin{enumerate}
    \item $D$ runs through a set of representatives of $GL_{n}(\mathbb{Z})\backslash \mathcal{D}(\Gamma_n M \Gamma_n);$
    \item $B$ runs through a set of representatives of $\mod D$ incogruent matrices in $\mathcal{B}(D,\Gamma_n M \Gamma_n).$
\end{enumerate}
\end{thm}
\begin{proof}
See \cite[Theorem 3.4, Chapter VI]{krieg1990hecke}
\end{proof}

\begin{prop} \label{label 21}
Let $p$ be a prime number and $l\in \mathbb{N}$. Then the set $S_n(p^l)$ decomposes into finitely many right cosets relative to $\Gamma_n$. A set of representatives is given by 
\begin{equation}
 \mqty(p^l (D^t)^{-1} & B \\ 0 & D) \label{label 9}
\end{equation}
where $D$ runs through a set of representatives of
\begin{equation*}
   GL_{n}(\mathbb{Z}) \backslash \Big\{D\in \textrm{Mat}_n (\mathbb Z)\mid \textrm{Elm}(D)=\textrm{diag}(d_1,d_2,\dots,d_n)\,\,\, \textmd{and}\,\,\, d_i \mid p^l \,\,\, \textmd{for all}\,\,\, i=1,2,\dots,n \Big\}
\end{equation*}
and $B$ runs through a set of representatives of $\mod D$ incongruent matrices in 
\begin{equation}
    B(D):= \{B\in \textrm{Mat}_n(\mathbb{Z}) \mid B^t D  = D^t B\} \label{label 18}
\end{equation} 
\end{prop}

\begin{proof}
Every right coset $\Gamma_n M$ contains a representatives of the form in \eqref{label 9}. Suppose we have two representatives $N$ and $M$ of this form:
\begin{equation*}
    N= \mqty(A & B \\ 0 & D), \quad M= \mqty(A' & B' \\ 0 & D'),
\end{equation*}
with $\Gamma_n N  =  \Gamma_n N $. Then since 

\begin{equation}\label{label 10} 
    \Big \{\mqty(A & B \\ 0 & D) \in \Gamma_n\Big\}=\Big\{\mqty(U^t & 0 \\ 0 & U^{-1}) \mqty ( 1_n & S \\ 0 & 1_n ) \mid U\in GL_{n}(\mathbb{Z})\,, S\in \textrm{Sym}_n(\mathbb{Z})\Big \},
\end{equation}

we obtain from the conditions in \eqref{label 7}-\eqref{label 8} that there exists $U \in GL_{n}(\mathbb{Z})$ such that $ D(A')^t = \lambda(M)U^{-1}$ and hence
\begin{equation*}
    D (D')^{-1} = U^{-1}.
\end{equation*}
This shows that $D=D'$ and consequently $A=A'$. Moreover from the equality $\eqref{label 10}$ we know that there exists $S\in \textrm{Sym}_n(\mathbb{Z})$ such that
\begin{align*}
    -A(B')^t + B (A)^t 
    & = -A(B')^t + AB^t, \\
    & = \lambda(M)S.
\end{align*}
Since $A^t D = \lambda(M)=\lambda(N)$ we obtain that $B-B'=SD$, i.e $B \equiv B' \mod D$.
\end{proof}

\begin{lem} \label{label 22}
Let $p$ be a prime number, $l \in \mathbb{N}$, $D \in  \textrm{Mat}_n (\mathbb{Z})$ such that $\textrm{Elm}(D)=\textrm{diag}(d_1,d_2,\dots,d_n)$ with $d_i | p^l$ for $i=1,2,\dots,n$ 
and $B(D)$ is as in equation \eqref{label 18}. Given $U,V \in GL_{n} (\mathbb Z)$, then

\begin{enumerate}
    \item $ |B(D) \mod D|=|B(UDV) \mod UDV |$,
    \item $|B(D) \mod D|=d_1^{n}d_2^{n-1}\dots d_n$.
\end{enumerate}
\end{lem}

\begin{proof}
Since $U,V \in \mathcal{U}_n (\mathbb{Z})$, we have the following bijection
\begin{align*}
    B(D) & \rightarrow B(UDV)\\
     M & \mapsto (U^t)^{-1} MV
\end{align*} 
This proves the first claim. Hence we can suppose that $D=\textrm{Elm}(D)$ to prove the second assertion. Since $d_1 | \dots | d_n$, we can write
    \begin{equation*}
        B(D)=\{M={(b_{jk}) \mid b_{jk}\in \mathbb{Z}}, b_{jk}= b_{kj} \frac{d_k}{d_l}\,\,\, \textmd{for}\,\,\, j\leq k\}.
    \end{equation*}
    By definition of the relation in $\eqref{label 19}$ the entries $b_{jk}$ may be reduced $\mod d_k$. This shows that $B(D)$ consists exactly of $d_1^{n}d_2^{n-1}\dots d_n$ equivalence classes $\mod D$.
\end{proof}

Let $r\in \mathbb{N}$ and define

\begin{equation*}
    R_{\Gamma_n} (r):= \sum_{g \in \Gamma_{n} \backslash S_n (r) / \Gamma_{n}   } \deg_{\Gamma_n}(g).
\end{equation*}

\begin{prop} \label{label 54}
The function $ R_{\Gamma_n}: \mathbb{N} \rightarrow \mathbb{N}$ is a multiplicative function, i.e given relatively prime numbers $q,r\in \mathbb{N}$, we have
\begin{equation*}
     R_{\Gamma_n}(qr)= R_{\Gamma_n}(q) R_{\Gamma_n}(r).
\end{equation*}
\end{prop}

\begin{proof}
Observe first that if $g\in S_{n}(qr)$, then $\textrm{Elm} (g)=\textrm{Elm}(g_1)\textrm{Elm}(g_2)$ for some $g_1\in S_n(q)$ and $g_2 \in S_n(r)$. Hence

\begin{equation*}
    R_{\Gamma_n}(qr)=  \sum_{g \in \Gamma_{n} \backslash S_n (qr) / \Gamma_{n}   } \deg_{\Gamma_n}(g) = 
    \sum_{g_2 \in \Gamma_{n} \backslash S_n (q) / \Gamma_{n}   }
    \sum_{g_1 \in \Gamma_{n} \backslash S_n (r) / \Gamma_{n}   } \deg_{\Gamma_n}(g_1g_2),
\end{equation*}
so it is enough to show that $\deg(g_1g_2)=\deg(g_1)\deg(g_2)$. We decompose the double cosets $\Gamma_n \backslash g_1 / \Gamma_n$ and $\Gamma_n \backslash g_2 / \Gamma_n$ into finitely many right cosets $\Gamma_n Q_i, i=1,\dots \deg (g_1)$ and $\Gamma_n R_i, i=1,\dots \deg (g_2)$ and consider the right cosets given by $\Gamma Q_iR_j$. Suppose hat $\Gamma Q_iR_j=\Gamma Q_kR_l$. Then there exists some $\gamma\in \Gamma_n$ and a matrix $M \in GSp_{2n}(\mathbb Q)$ such that

\begin{equation*}
    M= Q_k^{-1}\gamma Q_i=R_l R_j^{-1},\quad \lambda(Q_k)=\lambda(Q_i)=q,\quad 
    \lambda(R_k)=\lambda(R_i)=r.
\end{equation*}
Recall from \eqref{label 17} that $Q_k=\lambda(Q_k)^{-1}\Omega^{-1} Q_k^t \Omega$ and so after writing $M=\{\frac{n_{ij}}{m_{ij}}\}_{ij}$ where $(n_{ij},m_{ij})=1$, we see that the integers $m_{ij}$ divide $\lambda(Q_k)=q$ and similarly $m_{ij}$ divide $\lambda(R_j)=r$. By assumption $(q,r)=1$ so $M\in \Gamma_n$ since $\lambda(A)=1$. This shows that $\Gamma_n Q_k=\Gamma_n Q_i$ and $\Gamma_n R_l=\Gamma_n R_k$, in other words $i=k$ and $l=j$. To conclude we simply observe that the cosets $\Gamma Q_iR_j$ form a partition of $\Gamma_n \backslash g_1 g_2 / \Gamma_n$.
\end{proof}

Let $p$ be a prime and $l\in \mathbb N$. For arbitrary $n$, it is in general not possible to obtain a closed formula for $\deg(a)$ if $a\in S_{p^l}$ is given in its elementary form. On the other hand, an upper bound of $\deg(a)$ will be enough in most of our calculations.
We first suppose that $a$ is given by

\begin{equation*}
    a=\textmd{diag} (p^{k_1},p^{k_2},\dots,p^{k_n},p^{l-k_1},\dots,p^{l-k_n}),\quad [l/2]\leq k_1\leq k_2\leq \dots \leq k_n
\end{equation*}

Note that this is \textit{not} the elementary sympletic form of $a$ since $k_1\geq [l/2]$.  We shall use the root datum of $GSp_{2n}$ given in Proposition \ref{label 37}. The set $\Phi^+$  of positive roots is given by (see \cite[page 167]{tadic1994representations})

\begin{align*}
e_j-e_i,\quad 1 \leq i< j \leq n\\
e_j+e_i -e_0,\quad 1 \leq i< j \leq n\\
2e_i - e_0, \quad 1 \leq i \leq n.
\end{align*}

where we have used our choice of the basis $e_i,i=1,\dots,n$ (the choice of the basis used in \cite{tadic1994representations} is different but the computations are essentially the same). Hence
\begin{equation*}
    2\rho = \sum _{\alpha \in \Phi^+} \alpha =  2 \sum_{i=0}^{i=n-1} (n-i)e_{n-i} - \frac{1}{2}n(n+1)e_0.
\end{equation*}

We set

\begin{equation*}
     \lambda = \sum_{i=1}^{i=n}k_i f_i + l f_0.
\end{equation*}

Observe that  $\langle \lambda,\alpha \rangle \geq 0 $ for all $\alpha \in \Phi^+$. Then using the degree formula in Proposition $7.4$ in \cite{gross1998satake} (see also the proof of Corollary $1.9$ in \cite{clozel2001hecke} for a similar result in the case of $GL_n$) we obtain

\begin{equation}
    \deg (a) = p^{(\sum_{i=0}^{i=n-1}2(n-i)k_{n-i})-\frac{1}{2}n(n+1)l} (1+O(p^{-1})) \label{label 38}
\end{equation}
where the big $O$ depends only on $n$. If $a \in S_{p^l}$ is given in its elementary symplectic form \eqref{label 39} , we apply left and right (symplectic) permutations matrices to $a$ (which leaves invariant the degree) and use the formula \eqref{label 38}.

\subsection{ Structure theorems of the symplectic group}

As seen in the previous section, the class number $h(G,K)$ plays a detrimental role in the definition of the abstract BCM system. The aim of this section is to show that for $G=GSp_{2n}$ we have that $h(G,K)=1$ where $K$ is any open compact subgroup of $GSp_{4}(\hZ)$. The proof relies on the notion of strong approximation in algebraic groups, which we briefly review.  Given a linear reductive group $G$ over a global field $K$ and a non-empty finite set $S$ of places of $K$, we denote by $\mathbb A _{S}$ the ring of $S$-ad\`eles and let $G(\mathbb A_S)$ be the ring of $S$-ad\`eles of $G$

$$G (\mathbb A_{S}):= \{g=(g_{\upsilon}) \in \prod_{ \upsilon \notin S} G(K_\upsilon) \mid g_{\upsilon} \in G(\mathcal{O}_{\upsilon})\,\,\, \textmd{for almost all}\,\,\, \upsilon \notin S\}.$$

For any given $S$, there is a canonical embedding $G(K) \hookrightarrow G(\mathbb A_{S})$. An algebraic group $G$ over a global field $K$ has the strong approximation with respect to $S$ if $G(K)$ is dense in $G(\mathbb A_S)$. It is well known that strong approximation does not hold in general (for example one can take the group $G=GL_{n}$ with $S=\{\infty\}$). The following theorem provides a necessary and sufficient condition for the strong approximation theorem to hold for algebraic groups.

\begin{thm} (See {\cite{kneser1965starke}, \cite{platonov1969problem} in characteristic zero} and \cite{prasad1977strong},\cite{margulis1977cobounded}  \cite{margulis1991discrete} in positive characteristic), \label{label 73}
Let $G$ be an absolutely almost simple simply connected algebraic group over a field $K$ and $S$ a finite nonempty set of places of $K$. Then $G$ has strong approximation with respect to $S$ if and only if the group $G_S=\prod_{\upsilon \in S} G(K_{\upsilon})$ is noncompact.
\end{thm}

We can conclude from Theorem \ref{label 73} that the symplectic group  $G=Sp_{2n}$ (over $\mathbb Q$) has the strong approximation with respect to $S=\{\infty\}$. In the rest of this section, we provide an elementary proof of this result using matrix factorization. For this, we first put

\begin{equation}
    \Gamma_n(N)=\{\gamma \in \Gamma_{n} \mid \gamma^t\, \Omega\, \gamma \equiv \Omega\, \, \mod N \}, \label{label 13}
\end{equation}

for every positive integer $N$

\begin{lem}\label{label2}
Let $\pi_{N}$ be the projection $\pi_N: Sp_{2n}(\mathbb{Z}) \rightarrow Sp_{2n}(\mathbb{Z}/{N\mathbb{Z}})$ defined by $\pi_{N}(\gamma)=\gamma\mod N$. Then the sequence 
\begin{equation*}
    1 \longrightarrow \Gamma_n(N) \longrightarrow \Gamma_n \overset{\pi_N} \longrightarrow Sp_{2n}(\mathbb{Z}/{N\mathbb{Z}}) \longrightarrow 1
\end{equation*}
is exact.
\end{lem}

\begin{proof}
 The only nontrivial part is the surjectivity of the map $\pi_N$. This follows directly from \cite[Theorem 1]{newman1964symplectic}
\end{proof}

We denote by $\hZ=\varprojlim_{N>1} \mathbb{Z}/N\mathbb{Z}$ the ring of profinite integers and we write $M_{2n}(\hZ)$ for the ring of $2n\times 2n$-matrices with coefficients in $\hZ$. The profinite compact group $GL_{2n}(\hZ)=\varprojlim_{N>1} GL_{2n}(\mathbb{Z}/N\mathbb{Z})$ is a subset of $M_{2n}(\hZ)$ and consists of invertible matrices. The subgroup $Sp_{2n}(\hZ)\triangleleft GL_{2n}(\hZ)$ is defined by exactness of the sequence \eqref{label1}.

\begin{prop}\label{label 41}
$Sp_{2n}(\Z)$ is dense in $Sp_{2n}(\hZ)$.
\end{prop}

\begin{proof}
Since $M_{2n}(\hZ)=\varprojlim_{N>1} M_{2n}(\mathbb{Z}/N\mathbb{Z})$ is a profinite ring, a system of neighborhoods of the zero matrix is given by  $\{NM_{2n}(\hZ) \mid N\in \mathbb{N}\}$ and thus a system of neighborhood of $\mathbb{1}_{2n}$ in $Sp_{2n}(\hZ)$) is given by $\{U_N \mid N\in \mathbb{N}\}$ where $U_N= \{1+ NM_{2n}(\hZ)\}\cap Sp_{2n}(\hZ)\}$. Given $N\in \mathbb{N}$, consider the projection map $\pi_N: Sp_{2n}(\hZ) \rightarrow Sp_{2n}(\Z/N\mathbb{Z})$ and note that $\ker{\pi_N}=U_N$. By Lemma \ref{label2} the projection $\pi_N$ restricted to ${Sp_{2n}(\Z)}$ is surjective. Hence for any given $x\in Sp_{2n}(\hZ)$ and $N\in \mathbb{N}$, we choose $\gamma_N\in Sp_{2n}(\mathbb{Z})$ such that $\pi_N(\gamma_N)=\pi_N(x)$.Then $x^{-1}\gamma_N \in U_N$, that is $\gamma_N\in xU_N$. Taking $N$ large enough shows that $Sp_{2n}(\Z)$ is dense in $Sp_{2n}(\mathbb{\hZ})$.
\end{proof}

For a prime $p$, we denote by  $\mathbb{Q}_p$ the field of $p$-adic numbers and $\mathbb{Z}_p$ its compact subgring of $p$-adic integers.  We consider $\mathbb{A}_{\mathbb Q,f}$ the ring of finite ad\`eles of $\mathbb{Q} $, that is, the restricted product of the fields $\mathbb{Q}_p$ with respect to $\mathbb{Z}_p$ and we denote by $I_\mathbb Q=\mathbb A_{\mathbb Q,f}^{\times}$ the id\`ele group. 

\begin{thm} \label{label 3}
The algebraic group $G=Sp_{2n}$ (over $\mathbb Q$) has the strong approximation with respect to the infinite place $S=\{\infty\}$. 
\end{thm}
\begin{proof}
We denote by $H$ the closure of $Sp_{2n}(\mathbb{Q})$ in $Sp_{2n}(\mathbb{A}_{\mathbb Q ,f})$. We have a dense diagonal embedding of $\mathbb Q$ inside $\mathbb A _{\mathbb Q,f}$, whence the subgroup $H$ contains the group generated by the matrices of the form \eqref{label 27} with $\alpha_1,\dots,\alpha_5 \in \mathbb A _{\mathbb Q ,f}$. In particular, given any prime number $p$, the subgroup $H$ contains the set of matrices of the form \eqref{label 27} with $(\alpha_i)_q=1$ for $q\neq p$ and $i=1,\dots,5$. Since $Sp_{4}(\mathbb Q_p)$ is generated by these type of matrices, we see that $H$ contains the elements $M = (M_p)_p\in Sp_{2n}(\mathbb{A}_{\mathbb Q ,f})$, with $M_p\in Sp_{2n}(\mathbb Q_p)$ and $M_q=1$ for $q\neq p$. Hence for any finite set of primes $F$, we have the inclusion

\begin{equation*}
    \{(x)_q \in Sp_{2n}(\mathbb{A}_{\mathbb Q ,f}) \mid \forall q\in F, x_{q} \in Sp_{2n}(\mathbb Q_p) \,\,\, \textmd{and}\,\,\, x_{q}=1  \,\,\,\textmd{if}\,\,\,q \notin F \} \subseteq H.
\end{equation*}
The result follows since the union of these subsets over all finite set of primes $F$ is dense in $Sp_4(\mathbb{A}_{\mathbb Q,f})$.
\end{proof}

\begin{cor}\label{label4}
Let $K$ be an open compact subgroup of $GSp_{2n}(\mathbb A_f)$. Then $\lambda(K)\subseteq \hZ^\times $ and if $\lambda(K)=\hZ^\times $, we have
\begin{equation} \label{label 95}
    GSp_{2n}(\mathbb A_f)= K \cdot GSp_{2n}^+(\mathbb Q)= GSp_{2n}^+(\mathbb Q) \cdot K\,\,\,.
\end{equation}
In particular we get that $h(G,K)=1$ for the maximal open compact subgroup $K=GSp_{2n}(\hZ)$.
\end{cor}

\begin{proof}
The first assertion follows from the fact that the map $\lambda: GSp_{2n}(\mathbb A_f) \rightarrow \mathbb A_f^{\times}$ is continuous and $\mathbb Z_p^\times$ is the unique maximal subgroup of $\mathbb Q_p^\times$. To show \eqref{label 95}, let $g\in GSp_{2n}(\mathbb{A}_f)$ so that $\lambda(g)\in \mathbb{A}_f^{\times}$. Since $\mathbb A_{\mathbb Q,f}^\times = \mathbb Q^{\times} \hZ^{\times}$ , we can write 
\begin{equation*}
    \lambda(g)=\alpha \cdot x,
\end{equation*}
for some $\alpha \in \mathbb{Q}^{\times}$ (we choose $\alpha >0$ if necessary) and $x\in \hZ^{\times}=\lambda(K)$. Thus we can choose $k\in K$ such that $\lambda(k)=x$.  Consider the matrix

\begin{equation*}
    g'=\textrm{diag}(\underbrace{\alpha^{-1},1,\dots, \alpha^{-1},1}_n,\underbrace{1,\alpha^{-1},\dots,1,\alpha^{-1}}_n)\, g  k^{-1}.
\end{equation*}

\par Observe that $g'\in Sp_{2n}(\mathbb{A}_f)$ and by Theorem \ref{label 3} the open set $g'\cdot Sp_{2n}(\hZ) \subseteq Sp_{2n}(\mathbb{A}_f)$ contains $\eta \in Sp_{2n}(\mathbb{Q})$ such that $\eta = g' \cdot h$ for some $h\in Sp_{2n}(\hZ)$. Moreover by Proposition \ref{label 41} the group $Sp_{2n}(\mathbb Z)$ is dense in $Sp_{2n}(\hZ)$, hence we can find  $\gamma \in Sp_{2n}(\mathbb Z)$ such that 
$\gamma \in  g'^{-1}\eta U$, where $U=K\cap Sp_{2n}(\hZ)$. This shows that $g \in K \cdot GSp_{2n}^+(\mathbb{Q})$ as desired. Considering the automorphism $x \mapsto x^{-1}$ for $x\in GSp_{2n}(\mathbb{A}_f)$, we see that $GSp_{2n}(\mathbb{A}_f)= GSp_{2n}^+ (\mathbb{Q}) \cdot K$.
\end{proof}

\subsection{Siegel upper half plane}

\begin{definition}
The Siegel upper half plane of degree $n$ consists of all symmetric complex $n\times n$-matrices whose imaginary part is positive definite:
\begin{equation}
    \mathbb{H}_n^+=\{\tau=\tau_1+i\tau_2 \in \textmd{Mat}_n(\mathbb C)\mid \tau^t=\tau,  \quad \tau_2 > 0 \}
\end{equation}
\end{definition}

Let $M=\mqty(A & B \\ C & D) \in GSp_{2n}^+(\mathbb{R})$ and $\tau \in \mathbb{H}_n^+$. Then the matrix $C \tau+D$ is invertible and if we define 
$$g \cdot \tau := (A\tau+B)(C\tau+D)^{-1},$$

then the map 
\begin{equation*}\label{label 96}
\tau \mapsto g \cdot \tau 
\end{equation*}
is an action of $GSp_{2n}^+(\mathbb{R})$ on $\mathbb{H}_n^+$ \cite{pitale2240siegel} . If we write $\tau =\tau_1+i\tau_2 \in \mathbb{H}_n^+$ and $d\tau=d\tau_1d\tau_2$ is the Euclidean measure, then the element of volume on $\mathbb{H}_n^+$ given by
    \begin{equation} \label{label 98}
        d^*\tau:= \det(\tau_2)^{-(n+1)} d\tau,
    \end{equation}
    is invariant under all transformations of the group $GSp_{2n}^+(\mathbb{R})$, i.e
    \begin{equation*}
        d^* (g \cdot \tau)  = d^* \tau, \quad \textmd{for all}\, g\in GSp_{2n}^+(\mathbb{R}).
    \end{equation*}

Given an element $\tau=\tau_2+i\tau_2 \in \mathbb{H}_n^+$, the relation
\begin{equation*}
    \mqty( 1_n & \tau_1 \\ 0_n & 1_n ) \mqty ( \tau_2^{1/2} & 0_n \\ 0_n & \tau_2^{-1/2} )\cdot i1_n = \tau
\end{equation*}
shows that the action of $GSp_{2n}^+(\mathbb R)$ is transitive. The stabilizer of $i1_n$ is the subgroup
\begin{equation*}
S = \Big\{\mqty(A & B \\ -B & A)\in GL_{2n}(\mathbb{R}) \Big\} \cap GSp_{2n}^+(\mathbb{R}), \label{label 66}
\end{equation*}
Hence the group $Z(\mathbb R)\backslash GSp_{2n}^+(\mathbb R)$ (where $Z(\mathbb R)$ denotes the center of $GSp_{2n}^+(\mathbb R)$ ) acts transitively on $\mathbb H_n^+$ and we have the following identification 
\begin{equation} \label{label 97}
    \mathbb{H}_n^+= PGSp_{2n}^+(\mathbb{R}) / K,
\end{equation}

where $K$ is the compact group $K = Z(\mathbb R)\backslash S\simeq \mathbb U^n /\{\pm 1_{2n}\} $ and $\mathbb U^n = S\cap Sp_{2n}(\mathbb R)$ is isomorphic to the unitary group of order $n$ through the map $$\mqty(A & B \\ -B & A) \mapsto A + i B.$$ The restriction of the action defined in \eqref{label 96} to the arithmetic subgroup $Sp_{2n}(\mathbb Z)$ will be of special importance to us. Note that from the identification \eqref{label 97} we see 
that the action  of $\Gamma_{n}$ on $\mathbb H_n^+$ is proper since $\Gamma_{n} =Sp_{2n}(\mathbb Z)$ is discrete and the subgroup $K$ in \eqref{label 66} is compact. The action of $Sp_{2n}(\mathbb Z)$ on $\mathbb H_{n}^+$ admits a fundamental domain and thanks to a result by Siegel \cite{SymplGeomSiegel}, it has the following concrete description \cite{namikawa2006toroidal}. Let $U_n$ be the subset of matrices $\tau=\tau_1+i\tau_2 \in  \mathbb H_n^+$ satisfying the following conditions:

\begin{enumerate}
    \item $\abs{\det(C\tau+D)}\geq 1$ for every $\mqty(A & B \\ C & D) \in Sp_{2n}(\mathbb Z);$
    \item $\tau_2=(\tau_2)_{ij}$ is Minkowski reduced, i.e $$a^t\tau_2 a \geq (\tau_2)_{kk}, \,\,\, 1 \leq k\leq n \quad \text{for all}\quad  a=\mqty (a_1 \\ \\
    \vdots \\ a_n) \in \mathbb Z^{n}\,\,\, \text{where}\,\,\, (a_1,\dots, a_n)=1;$$
    \item $\abs{(\tau_1)_{ij}} \leq 1/2.$
\end{enumerate}

Then $U_n$ is a fundamental domain of $Sp_{2n}(\mathbb Z)$ on $\mathbb H_{n}^+$ of finite volume with the respect to the element of volume \eqref{label 98}:
$$
\textmd{vol}(U_n)= 2 \prod_{i=1}^{n}\pi^{-k}\Gamma(i)\zeta(2i),
$$
where $\Gamma(s)$ denotes the Gamma functions and $\zeta(s)$ is the Riemann zeta function.

\subsection{Bost-Connes-Marcolli system: the $GSp_{2n}$-case}
We consider the connected Shimura datum $(GSp_{2n}^+, \mathbb{H}^{+}_n)$ together with the BCM pair 

\begin{align*}
& G^+=GSp_{2n}^+,\quad  X^+=\mathbb{H}^{+}_n, \quad V=\mathbb{Q}^{2n}, \quad M=MSp_{2n}(\hZ),\\
& L=\mathbb{Z}^{2n}, \quad K=GSp_{2n}(\hZ), \quad K_M=MSp_{2n}(\hZ).
\end{align*}

Let $M\in GSp_{2n}^+(\mathbb{Q}) \cap GSp_{2n}(\hZ)$. Since $\hZ\, \cap \,\mathbb{Q}=\mathbb{Z}$ it is clear that $M\in MSp_{2n}(\mathbb{Z})$ and there exists $M'\in MSp_{2n}(\mathbb{Z})^+$ such that $MM'=1$ so that $\lambda(M)=1$, that is $\Gamma_+=GSp_{2n}^+(\mathbb{Q}) \cap GSp_{2n}(\hZ)=\Gamma_{n}$. Corollary \ref{label4} tells us that $h(G,K)=1$ so by Proposition \ref{label 60} it is enough to work with 
the space $\mathfrak{Z}_{\mathcal{D},\mathcal{L}}^{\textmd{+}}$. As in the case of the $GL_2$ system, the first difficulty that arises is the presence of points in $\mathbb H_{n}^{+}$ with nontriival stabilizers:

\begin{prop} \label{label 74}
The groupoid structure on $\mathfrak{U}_{\mathcal{D},\mathcal{L}}^{\textmd{+}}$ does not pass to the quotient by the action of $\Gamma_{n} \times \Gamma_{n} $.
\end{prop}
\begin{proof}
Let $g=\mqty(1_{n} & 0_n \\ 0_n & \frac{1}{2} 1_{n}) \in GSp_{2n}^+(\mathbb{Q})$ and assume the groupoid composition is defined when we pass to the quotient. Since $g\cdot \frac{1}{2} i 1_{2n}= i 1_{2n}$ we obtain that
\begin{equation*}
    (g,0,i 1_{2n}) (g,0,\frac{1}{2}i  1_{2n}) = (g^2,0,\frac{1}{2}i  1_{2n}),
\end{equation*}
where the equality holds in the quotient. On the other hand, let $\gamma=\mqty(0_n &- 1_{n} \\  1_{n} & 0_n)$ so that $g\gamma^{-1}g=-\frac{1}{2}\gamma$ and $\gamma \cdot i  1_{2n} = g\cdot \frac{1}{2}i 1_{2n}$. We then have the following equality in the quotient:
\begin{equation*}
    (g^2,0,\frac{1}{2}i 1_{2n})= (g\gamma^{-1},0,\gamma \cdot i 1_{2n}) (g,0,\frac{1}{2}i 1_{2n})=(g\gamma^{-1}g,0,\frac{1}{2}i 1_{2n}).
\end{equation*}
Hence there exist $\gamma_1,\gamma_2 \in \Gamma_n=Sp_{2n}(\mathbb{Z})$ satisfying the following two conditions:
\begin{equation*}
   \gamma_2\cdot i  1_{2n}=i 1_{2n},
\end{equation*}
\begin{equation*}
    \gamma_1 g^2 \gamma_2^{-1}=\frac{1}{2} \mqty(0_n & 1_{n} \\ - 1_{n} & 0_n).
\end{equation*}
The first condition implies that $\gamma_2$ is of the form $\gamma_2=\mqty(A & B \\ -B & A)$ while the second condition gives
\begin{equation*}
    \gamma_1=\mqty(\frac{1}{2}B & -2A \\ \frac{1}{2}A & 2B).
\end{equation*}
Note that since $\gamma_1\in Sp_{2n}(\mathbb{Z})$ we get
\begin{equation*}
    \mqty(\frac{1}{4} (BA-AB) & A^2+B^2 \\ -B^2-A^2 & 4 (BA-AB))=\Omega,
\end{equation*}
that is $I=4(A'^2+B'^2)$ for some $A',B' \in \textrm{Mat}_n(\mathbb{Z})$, which is a contradiction.
\end{proof}

\subsection{The $GSp_{4,\mathbb Q}$-system}

We restrict our attention to the case $n=2$ and fix the following notation. We let  $Y=\mathbb H^+_2\times MSp_{4}(\hZ)$ so that $X=GSp_{4}(\mathbb Q)^+Y= \mathbb H^+_2\times MSp_{4}(\mathbb A_{f,\mathbb Q})$ and $ \Gamma_2=Sp_{4}(\mathbb Z)$. As observed above the action of $\Gamma_2$ on $Y$ is not free and it turns out that the set of points in $Y$ with non-trivial stabilizers strictly contains $\mathbb H_2^+\times \{0_4\}$.
Let $$  F_Y=\{h\in MSp_4(\hZ) \mid \rank_{\Q_p} (h_p) \leq 2\quad \textmd{for all primes $p$} \}. $$

Then the action of $\Gamma_2$ on $\tilde{Y}=Y\backslash (\mathbb{H}_2^+\times F_Y)$ is free. To see this, we suppose that for some $\gamma \in \Gamma_2$ we have $\gamma \cdot \tau = \tau$ and $\gamma h_p=h_p$ for some $\tau \in \mathbb{H}_2^{+}$ and $h_p \in MSp_4(\Q_p)$ with $\rank_{\Q_p}(h_p) >2$ for some prime $p$. Then we can find $T \in GL_4(\Q_p)$ such that
\begin{equation*}
    T \gamma T^{-1}=\mqty(1 & 0 & 0 & x_1 \\ 0 & 1 & 0 & x_2 \\ 0 & 0 & 1 & x_3 \\ 0 & 0 & 0 & x_4),
\end{equation*}
for some $x_1 \dots, x_4\in \Q_p$. Since the entries of $\gamma$ are in $\Z$ we see that $x_4\in \Q$ and thus $C_{\Q,\gamma}(x)=(x-1)^3(x-x_4)$. On the other hand, since $\gamma$ fixes a point in  $\mathbb{H}_2^+$, then there exists $P\in Sp_4(\R)$ such that
\begin{equation*}
    P\gamma P^{-1}= \mqty(a_1 & b_1 \\ -b_1 & a_1) \odot \mqty(a_2 & b_2 \\ -b_2 & a_2) ,\quad a_i,b_i\in \R  \quad a_1^2+b_1^2=a_2^2+b_2^2=1,
\end{equation*}
So $C_{\mathbb{C},\gamma}(x)=(x-\lambda_1)(x-\overline{\lambda}_1)(x-\lambda_2)(x-\overline{\lambda}_2)$ where $\lambda_1=a_1+ib_1$ and $\lambda_2=a_2+ ib_2$. Hence  $\lambda_1=\lambda_2=1$ and $\gamma=1$. It is easy to see that the set of points in $Y$ with non-trivial stabilizers is strictly larger than $\mathbb H_2^+ \times \mathbf \{0\}$.

 The fact that $F_Y$ is not invariant under scalar matrices in $GSp_{4}^+(\mathbb Q)$ creates a new difficulty that was not present in the $GL_2$-system. For this reason, and for the purpose of $\textmd{KMS}_\beta$ analysis, instead of working with the quotient $\mathbb H_2^+= GSp_{4}^+(\mathbb R) / K$, we consider first the quotient $PGSp_{4}^+(\mathbb R)= GSp_{4}^+(\mathbb R) / Z(\mathbb R)$, where $Z(\mathbb R)$ is the center of the group $GSp_4^+(\mathbb R)$. From now on we refer to this system as the $GSp_4$-system and we call the original dynamical system (corresponding to the Shimura datum $(GSp_{4}^+,\mathbb H_2^+)$) the \textit{Connes-Marcolli $GSp_{4}$-system}. We will show later that the two systems have the same thermodynamical properties.  Since now $PGSp_{4}^+(\mathbb R)$ is a group, we get the following:

\begin{prop}\label{label 14}
For $\beta \neq 0$, there exists a correspondence between $\textmd{KMS}_\beta$ states on the $GSp_{4}$-system and $\Gamma_2$-invariant measures $\mu$ on $PGSp_4^+(\mathbb R) \times MSp_{4}(\mathbb A_{f,\mathbb Q})$ such that  
\begin{equation}
    \nu(\Gamma_2 \backslash PGSp_{4}^+(\mathbb R) \times MSp_{4}(\hZ))=1,\quad \mu (gB)= \lambda(g)^{-\beta} \mu(B)\label{label 67}
\end{equation}

for any $g \in GSp_{4}^+(\mathbb Q)$ and Borel compact sublet $B \subset PGSp_{4}^+(\mathbb R) \times MSp_{4}(\mathbb A_{f,\mathbb Q})$. Here $\nu$ denotes the measure on $ \Gamma_2 \backslash PGSp_{4}^+(\mathbb R) \times MSp_4(\mathbb A_{f,\mathbb Q})$ corresponding to $\mu$.
\end{prop}

\begin{proof}
Since the action of $\Gamma_2$ on $PGSp_4^+(\mathbb R)$ is free, Proposition \ref{label 65} applied to the group $G=GSp_{4}^+(\mathbb Q)$ and the spaces $X= PGSp_{4}^+(\mathbb R) \times MSp_4(\mathbb A_{f,\mathbb Q})$ and $Y=PGSp_{4}^+(\mathbb R) \times MSp_4(\hZ)$ gives a one-to-one correspondence between the set of $\textmd{KMS}_\beta$ states on $(\mathcal{A},\sigma)$ and $\Gamma_2$-invariant measures $\mu$ on $Y$ such that $\mu (gB)=\lambda(g)^{-\beta} \mu(B)$ if $gZ$ and $Z$ are measurable subsets of $Y$. The equality $MSp_{4}(\mathbb A_f)=GSp_{4}^+(\mathbb Q) MSp_{4}(\hZ)$ allows us to extend (\cite[Lemma 2.2]{laca2007phase}) this measure to a Radon measure on $X$ such that $\mu(gB)=\lambda(g)^{-\beta}\mu(B)$ for every Borel subset $B\subseteq X$. Since the algebra $\mathcal{A}$ is not unital, from the normalization condition \eqref{label 69} and equation \eqref{label 70} we obtain that $\nu$ is a probability measure on $\Gamma_2 \backslash Y$.
\end{proof}
From now on, we let $Y=PGSp_{4}^+(\mathbb R) \times MSp_{4}(\hZ)$ and $X=PGSp_4^+(\mathbb R) \times MSp_{4}(\mathbb A_{f,\mathbb Q})$. For $\beta >0$, we denote by $\mathcal{E}_{\beta}$ the set of Radon measures on $X$ satisfying the properties in Proposition \ref{label 14}. Note that the extremal $\textmd{KMS}_\beta$ states correspond to point mass measures.

\section{$\textmd{KMS}_{\beta}$ states analysis}
\subsection{High temperature region}
\hfill
\\

We begin the $\textmd{KMS}_{\beta}$ of the $GSp_{4}$-system analysis by first considering the high temperature region $0<\beta <3$. Our first goal is to show that the $GSp_{4}$-system constructed above does not admit a $\textmd{KMS}_{\beta}$ state for $0<\beta < 3$ with $\beta \notin \{1,2\}$. We first show some useful lemmas.

\begin{lem} \label{label 49}
Let $F$ be a finite set of prime numbers and $g=(g_p)_{p\in F} \in \prod_{p\in F}MSp_4 {(\mathbb Z_p)} \subset \prod_{p\in F}MSp_4 {(\mathbb Q_p)}$ with $\lambda(g_p)\neq 0$ for all $p\in F$. Then there exist $g_1 \in S_{F}$ and $g_2\in \prod_{p\in F}GSp_4 {(\mathbb Z_p)}$ such that $g=g_1g_2$.
\end{lem}

\begin{proof}
It follows from Corollary \ref{label4} that we can find $g_1\in GSp_{4}^+(\mathbb Q)$ and $g_2\in GSp_{4}(\mathbb Z_p)$ such that $g = g_1 g_2$ with $g_1 \in GSp_{4}(\mathbb{Z}_q)$, $q\neq p$ and $g_1 \in MSp_{4}(\mathbb Z_p)$ and $\lambda(g_1) \in \mathbb N_F$, that is $g_1 \in S_F$.
\end{proof}

For  $k_0,k_1,k_2\in \mathbb{Z}$, we set
\begin{equation*}
   P_{k_0}:=\mqty(0 & 0 \\ 0 & p^{k_0}) \odot 0_2 ,\quad P_{k_1,k_2}:= \mqty(0 & 0 \\ 0 & p^{k_1}) \odot \mqty(0 & 0 \\ 0 & p^{k_2}).
\end{equation*}
\begin{align}
    Z_{k_0}^{(0)} := Sp_4(\mathbb{Z}_p)\,P_{k_0}\, GSp_4(\mathbb{Z}_p), \quad 
    Z_{k_1,k_2}^{(1)} := Sp_4(\mathbb{Z}_p)\, P_{k_1,k_2}\, GSp_4(\mathbb{Z}_p).
     \label{label 40}
\end{align}

\begin{lem} \label{label 46}
The sets $Z_{k_0}^{(0)}$ and $Z_{k_1,k_2}^{(1)}$, $k_0,k_1,k_2 \in \mathbb{Z}$ are pairwise disjoint. Moreover, given any nonzero matrix $a \in MSp_{4}(\mathbb{Q}_p)$ with $\lambda(a)=0$, then $a\in Z_{k_0}^{(0)} \cup  Z_{k_1,k_2}^{(1)} $ for some $k_0,k_1,k_2 \in \mathbb{Z}$.
\end{lem}

\begin{proof}
We first fix some notations. For  $1 \leq i,j \leq 4$, let $E_{ij}$ be the elementary matrix with coefficient $1$ at the position $(i,j)$ and $0$ otherwise. For $U \in GL_2(\mathbb{Z}_p)$ and $S \in \textrm{Sym}_2(\mathbb{Z}_p)$, we put 

\begin{equation*}
    J(U)  = \mqty(U^t & 0_2 \\ 0_2 & U^{-1}),\quad J(S)= \mqty(1_2 & S \\ 0_2 & 1_2)
\end{equation*}

Consider $g\in MSp_{4}(\mathbb{Q}_p)$ with $\mu(g)=0$. Let $g_0$ be any entry of $g$ with maximal $p$-adic valuation and we write $g_0=a_0 p^{k_0}$, where $a_0 \in \mathbb{Z}_p^{\times}$. Using the matrices $\Omega$ and $J_1(P)$ (where $P$ is a permutation matrix), we may assume that $g_{11}=g_0$. If $a$ is an entry of the matrix $g$, we set
\begin{equation*}
   U_a=1_2 - g_0^{-1} a E_{21},\quad S_a=-g_0^{-1} a E_{11},\quad \tilde{S}_a= -g_0^{-1} a (E_{12}+ E_{21})
\end{equation*}
Observe that by maximality, these matrices are in $Sp_4(\mathbb{Z}_p)$. We multiply $g$ from the right by
\begin{equation*}
    J(U_{g_{_{12}}}) J(S_{ g_{_{13}}+ g_{_{11}} g_{_{12}}g_{_{14}}}) J(\tilde{S}_{g_{_{14}}}).
\end{equation*}
to obtain a matrix whose first row is $g_0 \mathbf{e_1}$. Taking the transpose and repeating this process, we obtain a matrix whose first column is equal to $g_0\mathbf{e_1}^t$. The symplectic relations \ref{label 7} and \ref{label 8} imply that this matrix has the following form:

\begin{equation*}
    \mqty(g_0 & 0 \\ 0 & 0 ) \odot M, \quad M\in \textrm{Mat}_2(\mathbb{Q}_p),\quad \det(M)=0.
\end{equation*}
If $M=0$, then one has

\begin{equation*}
    \mqty(0 & -1 \\ 1 & 0) \odot 1_{2} \cdot  \mqty(g_0 & 0 \\ 0 & 0 ) \odot M \cdot 
    \mqty(0 & a_0^{-1} \\ 1 & 0) \odot \mqty(0 & a_0^{-1} \\ 1 & 0)=\mqty(0 & 0 \\ 0 & p^{k_0}) \odot 0_{2} \in P_{k_0}.
\end{equation*}
Otherwise, we can use right and left multiplication to find $\gamma_1\in SL_2 (\mathbb{Z}_p)$ and $\gamma_2 \in GL_2(\mathbb{Z}_p)$ such that $\gamma_1 M \gamma_2=\mqty(0 & 0 \\ 0 & p^{k_1})$, for some $k_1\in \mathbb{Z}$. Then the matrix

\begin{equation*}
 \mqty(0 & -1 \\ 1 & 0) \odot \gamma_1 \cdot \mqty(g_0 & 0 \\ 0 & 0 ) \odot M 
 \cdot  \mqty(0 & a_0^{-1} \\ -a_o \det(\gamma_2) & 0) \odot \gamma_2 =  \mqty(0 & 0 \\ 0 & p^{k_0}) \odot \mqty(0 & 0 \\ 0 & p^{k_1}) \in P_{k_0,k_1}
\end{equation*}
has the desired form. One can easily check that this decomposition is unique.
\end{proof}

\begin{lem} \label{label 83}
Let $p$ be a prime number and $g\in MSp_{4}(\mathbb Z_p)$ such that $\abs{\lambda(g)}_p =p^{-k},k\in \mathbb N$.  Then there exist $\gamma_1, \gamma_2 \in Sp_4(\mathbb Z _p)$ such that $\gamma_1 g\gamma_2$ is of the form 
\begin{equation*}
    \textmd{diag}(a_1,a_2,d_1,d_2), \quad \abs{a_1}_p \geq \abs{a_2}_p \geq \abs{d_2}_p \geq \abs{d_1}_p
\end{equation*}
\end{lem}

\begin{proof}
The proof is similar to \cite[Theorem 2.2, Chapter V]{krieg1990hecke}.
\end{proof}

\begin{lem}\label{label 12}
Let $p$ be a prime and we put
\begin{equation*}
    g_{1,p}:= \textrm{diag}(1,1,p,p),\quad
    g_{2,p}:= \textrm{diag}(p,p,p,p),\quad
    g_{3,p}:= \textrm{diag}(1,p,p^2,p).    
\end{equation*}
A set of representatives of the right cosets relative to $\Gamma_2$ in $\Gamma_2 g_{1,p}\Gamma_2$ is given by the matrices
\begin{equation*}
    \mqty(p & 0 & 0 & 0 \\
          0 & p & 0 & 0 \\
          0 & 0 & 1 & 0 \\
          0 & 0 & 0 & 1),\quad
    \mqty (p & 0 & 0 & 0 \\
           0 & 1 & 0 & k_1\\
           0 & 0 & 1 & 0 \\
           0 & 0 & 0 & p),\quad
    \mqty (1 & -k_2 & k_3 & 0 \\
           0 & p & 0 & 0\\
           0 & 0 & p & 0 \\
           0 & 0 & k_2 & 1),\quad
    \mqty (1 & 0 & k_4 & k_5 \\
           0 & 1 & k_5 & k_6\\
           0 & 0 & p & 0 \\
           0 & 0 & 0 & p),\quad
\end{equation*}
where $0 \leq k_1,k_2,\dots,k_6 < p$.

\par \noindent A set of representatives of the right cosets relative to $\Gamma_2$ in $\Gamma_2 g_{3,p}\Gamma_2$ is given by the matrices
\begin{align*}
& \mqty(p^2 & 0 & 0 & 0 \\
       0  & p & 0 & 0 \\
       0  & 0 & 1 & 0  \\
       0  & 0 & 0 & p ), \quad 
\mqty( p  & -pr_1  & 0 & 0 \\
       0  & p^2 & 0 & 0 \\
       0  & 0 & p & 0  \\
       0  & 0 & r_1 & 1 ), \quad
\mqty( p  & 0 & 0 &  pr_2 \\
       0  & 1 & r_2 & r_3 \\
       0  & 0 & p & 0  \\
       0  & 0 & 0 & p^2 ), \quad \\
& \mqty( 1  & -r_4 & r_5r_4+r_6 &  r_5 \\
       0  & p & pr_5  & 0 \\
       0  & 0 & p^2 & 0  \\
       0  & 0 & pr_7 & p ),\quad
  \mqty( p  & 0 & r_8 & r_9 \\
       0  & p & r_9 & r_{10} \\
       0  & 0 & p & 0  \\
       0  & 0 & 0 & p ), \quad 
\end{align*}
where $1 \leq r_1,r_2,r_4,r_5 < p$, $1 \leq r_3,r_6 < p^2$, and $0 \leq r_8,r_9, r_{10} < p$ are such that $r_p\mqty(r_8 & r_9 \\ r_9 & r_{10})=1$, where $r_p(B)$ denotes the rank of the matrix $B\in \textrm{Mat}_2(\mathbb{Z})$ over $\mathbb{Z}/p\mathbb{Z}$. In particular we have:
\begin{align*}
    \deg_{\Gamma_2}(g_{1,p}) & = (1+p)(1+p^2)\\
    \deg_{\Gamma_2}(g_{2,p}) & = 1 \\
    \deg_{\Gamma_2}(g_{3,p})& = p + p^2 + p^3 + p^4. 
\end{align*}
\end{lem}

\begin{proof}

Recall that $S_4(p)$ denotes the set of matrices $g\in GSp_4(\mathbb{Z})$ such that $\mu(g)=p$. This set consists of a single coset:
\begin{equation*}
    S_4(p)=\Gamma g_{1,p} \Gamma.
\end{equation*}
From this we can see that

\begin{equation*}
    \mathcal{D}(\Gamma_2 g_{1,p}\Gamma_2)=\{ \Gamma_1 \mqty(1 & 0 \\ 0 & 1) \Gamma_1 , \Gamma_1 \mqty(1 & 0 \\ 0 & p) \Gamma_1 , \Gamma_1 \mqty(p & 0 \\ 0 & p)\Gamma_1 \} 
\end{equation*}
The decomposition of $\Gamma g_{1,p} \Gamma$ into right cosets follows then from applying from applying Theorem \ref{label 10} with $n=2$ and $n=1$ (notice that we are using the convention that for $n$ impair, the matrices $B$ are under the diagonal).
\par The decomposition of  $\Gamma g_{2,p} \Gamma$ is trivial. To decompose the double cosets $\Gamma g_{3,p} \Gamma$, we use the following simple criterion:

\begin{equation*}
    M \in \Gamma_2 g_{3,p} \Gamma_2 \Leftrightarrow r_p(M)=1\, \textmd{and}\, M\in GSp_4(\mathbb{Z}).
\end{equation*}

We then obtain

\begin{equation*}
     \mathcal{D}(\Gamma_2 g_{3,p}\Gamma_2)= \{ \Gamma_1 \mqty(1 & 0 \\ 0 & p) \Gamma_1 , \Gamma_1 \mqty(p & 0 \\ 0 & p) \Gamma_1 , \Gamma_1 \mqty(p & 0 \\ 0 & p^2)\Gamma_1 \}.
\end{equation*}
For each $D \in U_2\backslash \mathcal{D}(\Gamma_2 g_{3,p}\Gamma_2)$, the set $\mathcal{B}(D,\Gamma_2 g_{3,p}\Gamma_2)$ is then obtained by applying again Theorem \ref{label 10} and using the relations
\begin{equation*}
    A^tD=p^2 \mathbb{1}_2,\quad B^tD=D^tB,\quad AB^t=B A^t.
\end{equation*}

\end{proof}

\begin{lem} \label{label 47}
Let $p$ be a prime and denote by $G_p$ the subgroup of $GSp_4^+(\mathbb{Q})$ generated by $\Gamma_2$ and the matrices $g_{1,p},g_{2,p}$ and $g_{3,p}$ defined in Lemma \ref{label 12}. Suppose that $\mu_p$ is a $\Gamma_2$-invariant measure on $PGSp_{4}^+(\mathbb R) \times MSp_4(\mathbb{Q}_p)$ such that
\begin{enumerate}
    \item $\mu_p(PGSp_{4}^+(\mathbb R) \times MSp_4(\mathbb{Z}_p)) < \infty$
    \item $\mu_p(gZ)=\lambda(g)^{-\beta} \mu_p(Z)$ for all $g\in G_p$ and Borel $Z\subseteq PGSp_{4}^+(\mathbb R)\times \mathbb MSp_4(\mathbb{Q}_p)$.
\end{enumerate}
If $\beta \notin \{1,2,3\}$ then $PGSp_{4}^+(\mathbb R)\times GSp_4(\mathbb{Q}_p)$ is subset of full measure in $PGSp_{4}^+(\mathbb R)\times MSp_4(\mathbb{Q}_p)$
\end{lem}

\begin{proof}
We define a measure $\tilde{\mu}_p$ on $MSp_{4}(\mathbb Q_p)$ by $\tilde{\mu}_p(Z)=\mu_p(PGSp_{4}^+(\mathbb R) \times Z)$. Note that by assumption we have $\tilde{\mu}_p(MSp_{4}(\mathbb Z_p))<\infty$. For any $g\in G_p$ and any positive integrable $\Gamma_2$-invariant function on $MSp_4{\mathbb ( Q_p)}$ we get from the second condition that

\begin{equation}
    \int_{MSp_{4}(\mathbb Q_{p})} T_{g} f \d \tilde{\mu}_ p = \lambda(g)^{\beta} \int_{MSp_{4}(\mathbb Q_{p})} f \d \tilde{\mu}_ p. \label{label 45}
\end{equation}

For $k_0,k_1,k_2 \in \mathbb Z$, consider the functions $f_{k_0}^{(0)}=I_{Z_{k_0}^{(0)}}$ and $f_{k_1,k_2}^{(1)}=I_{Z_{k_1,k_2}^{(1)}}$, where the sets $Z_{k_0}^{(0)}$ and $Z_{k_1,k_2}^{(1)}$ are as in equation \eqref{label 40}. Given $g\in GSp_{4}^+(\mathbb Q)$, we have that the function $T_g f_{0,0}^{(1)}$ is continuous and $\Gamma$-invariant. By Proposition \ref{label 41} the group $\Gamma$ is dense in $Sp_{4}(\mathbb Z_p)$, whence $T_g f_{0,0}^{(1)}$ is left $Sp_{4}(\mathbb Z_p)$-invariant. For $k_1,k_2 \in \mathbb Z$, we can expand the expression $ T_{g_{2,p}^{-1}g_{1,p}}f_{0,0}^{(1)}( P_{k_1,k_2}) $ using the explicit representatives given in  Lemma \ref{label 12}:

\begin{align*}
    & \frac{1}{\deg (g_{2,p}^{-1} g_{1,p})} \Big(
    f_{0,0}^{(1)}(
    \mqty(0 & 0 & 0 & 0 \\
          0 & 0 & 0 & 0 \\
          0 & 0 & p^{k_1-1}& 0 \\  
          0 & 0 & 0 & p^{k_2 -1})) + 
    \sum_{k=0}^{p-1}
    f_{0,0}^{(1)}(
    \mqty(0 & 0 & 0 & 0 \\
          0 & 0 & 0 & kp^{k_2-1}\\
          0 & 0 & p^{k_1 -1} & 0 \\
          0 & 0 & 0 & p^{k_2}) ) \\
    & + \sum_{0  \leq k',a\leq p-1} f_{0,0}^{(1)}(
    \mqty (0 & 0 & k' p^{k_1-1} & 0 \\
           0 & 0 & 0 & 0\\
           0 & 0 & p^{k_1} & 0 \\
           0 & 0 & ap^{k_1-1} & p^{k_2 -1})) +
           \sum_{0 \leq a,b,n \leq p-1} 
           f_{0,0}^{(1)} (
           \mqty (0 & 0 & bp^{k_1-1} & n p^{k_2 -1 }  \\ 
                  0 & 0 & np^{k_1 -1} & bp^{k_2-1} \\
                  0 & 0 & p^{k_1} & 0 \\
                  0 & 0 & 0 & p^{k_2})  )
    \Big)
\end{align*}

Since $ T_{g_{2,p}^{-1}g_{1,p}}f_{0,0}^{(1)}(\gamma_1 P_{k_1,k_2}\gamma_{2}) 
    = T_{g_{2,p}^{-1}g_{1,p}}f_{0,0}^{(1)}( P_{k_1,k_2}\gamma_2)$ for $\gamma_1 \in Sp_{4}(\mathbb Z_p) $ and $\gamma_2 \in GSp_{4}(\mathbb Z_p)$, it follows that

\begin{align*}
    \deg (g_{2,p}^{-1} g_{1,p}) T_{g_{2,p}^{-1} g_{2,p}} f_{0,0}^{(1)}
    & = f_{1,1}^{(1)}+ f_{1,0}^{(1)}+(p-1)f_{1,1}^{(1)}+ f_{0,1}^{(1)}+ (p^2-1)f_{1,1}^{(1)} + (p-1)f_{1,0}^{(1)} \\
    & + (p-1)^2 f_{1,1}^{(1)} + f_{0,0}^{(1)} + (p-1)f_{1,1}^{(1)}+ (p-1)f_{0,1}^{(1)} +(p-1)^2 f_{1,1}^{(1)} \\
    & + (p-1)^2 f_{1,1}^{(1)} + (p-1)^3 f_{1,1}^{(1)}.
    \\
    & = f_{0,0}^{(1)} + p f_{1,0}^{(1)} + p f_{0,1}^{(1)} + (p^3+p^2-p)f_{1,1}^{(1)}.
\end{align*}

Similarly, the expansion of $T_{g_{2,p}^{-2}g_{3,p}} f_{0,1} (P_{k_1,k_2})$ is given by

\begin{align*}
    & \frac{1}{\deg (g_{2,p}^{-2}g_{3,p})} \Big( 
    f_{0,1} \mqty (0 & 0 & 0 & 0\\
                   0 & 0 & 0 & 0\\
                   0 & 0 & p^{k_1-2} & 0\\
                   0 & 0 & 0 & p^{k_2-1}) +
    \sum_{ 0 \leq a\leq p-1}  
    f_{0,1} \mqty (0 & 0 & 0 & 0\\
                   0 & 0 & 0 & 0\\
                   0 & 0 & p^{k_1-1} & 0\\
                   0 & 0 & ap^{k_1-2} & p^{k_2-2})\\
    & +  \sum_{\substack{0 \leq b\leq p-1\\ 0 \leq c\leq p^2 -1 }} 
    f_{0,1} \mqty (0 & 0 & 0 & b p^{k_2-1}\\
                   0 & 0 & bp^{k_1-2} & cp^{k_2-2}\\
                   0 & 0 & p^{k_1-1} & 0\\
                   0 & 0 & 0 & p^{k_2}) + \sum_{r_p \mqty (b_1 & b_2 \\ b_2 & b_3)=1}
    f_{0,1} \mqty (0 & 0 & b_1 p^{k_1-2} & b_2 p^{k_2-1}\\
                   0 & 0 & b_2 p^{k_1-2} & b_3 p^{k_2 -2}\\
                   0 & 0 & p^{k_1-1} & 0\\
                   0 & 0 & 0 & p^{k_2-1})\\
    & + \sum_{\substack{0 \leq l\leq p^2-1 \\ 0 \leq k,d\leq p-1}}
    f_{0,1} \mqty (0 & 0 & (kd+l)p^{k_1-2} & kp^{k_2-2}\\
                   0 & 0 & kp^{k_1-1} & 0\\
                   0 & 0 & p^{k_1} & 0\\
                   0 & 0 & dp^{k_1-1} & p^{k_2-1})
    \Big).
\end{align*}
Note that there are $p-1$ positive integers divisible by $p$ between $0$ and $p^2-1$. Similarly there are $p(p-1)^2$ positive integers divisible by $p$ of the form $kd+l$ where $0<k,d<p$ and $1<l<p^2$. Hence

\begin{align*}
    \deg (g_{2,p}^{-2}g_{3,p})T_{g_2^{-2}g_3}f_{0,0}^{(1)}
    & = f_{2,1}^{(1)} + f_{1,2}^{(1)} + (p-1)f_{2,2}^{(1)} + f_{1,0}^{(1)} + (p-1) f_{1,1}^{(1)} + (p^2-p)f_{1,2}^{(1)}\\
    & + (p-1)f_{2,1}^{(1)} + (p-1)^2 f_{2,1}^{(1)} + \Big((p-1)(p^2-1)- (p-1)^2 \Big) f_{2,2}^{(1)}\\
    & + (p-1)f_{1,2}^{(1)} + (p-1)f_{2,1}^{(1)} + (p-1)^2 f_{2,2}^{(1)} + f_{0,1}^{(1)} + (p-1)f_{1,2}^{(1)}\\
    & + (p-1)f_{1,1}^{(1)} + (p-1)^2 f_{2,2}^{(1)} + (p-1)f_{1,1}^{(1)} + \Big((p^2-1)-(p-1)\Big) f_{2,1}^{(1)} \\
    & + (p-1)^2 f_{1,2}^{(1)}+ \Big( (p^2-p) (p-1)  \Big) f_{2,2}^{(1)} + (p-1)^2 f_{1,1}^{(1)}\\
    & + \Big( (p^2-p) (p-1)  \Big) f_{2,1}^{(1)} + p(p-1)^2 f_{1,2}^{(1)} \\
    & + \Big((p-1)^2(p^2-1)- p(p-1)^2 \Big) f_{2,2}^{(1)}\\
    & = f_{0,1}^{(1)} + f_{1,0}^{(1)} + (p^2+p-2)f_{1,1}^{(1)} + p^3 f_{1,2}^{(1)} + p^3 f_{2,1}^{(1)} + (p^4-p^3) f_{2,2}^{(1)}.
\end{align*}

Similar computations lead to the following identities:
\begin{align*}
    \deg (g_{2,p}^{-1}) T_{g_{2,p}^{-1}}f_{0,0}^{(1)} 
    & = f_{1,1}^{(1)},\\
    \deg (g_{2,p}^{-2} g_{1,p}) T_{g_{2,p}^{-2} g_{1,p}}f_{0,0}^{(1)} 
    & =  f_{1,1}^{(1)} + p f_{2,1}^{(1)} + p f_{1,2}^{(1)} + (p^3+p^2-p)f_{2,2}^{(1)},\\
    \deg (g_{2,p}^{-2}) T_{g_{2,p}^{-2}}f_{0,0}^{(1)} 
    & = f_{2,2}^{(1)}.
\end{align*}
We then have
\begin{align*}
    f_{0,0}^{(1)} 
    & = \deg (g_{2,p}^{-1} g_{1,p}) T_{g_{2,p}^{-1} g_{1,p}} f_{0,0}^{(1)} - p \deg (g_{2,p}^{-1} g_{3,p}) T_{g_{2,p}^{-1} g_{3,p}} f_{0,0}^{(1)} - (p+p^3)  \deg (g_{2,p}^{-1}) T_{g_{2,p}^{-1}}f_{0,0}^{(1)}\\
    & + p^3 \deg (g_{2,p}^{-2} g_{1,p}) T_{g_{2,p}^{-2} g_{1,p}}f_{0,0}^{(1)} - p^6  \deg (g_{2,p}^{-2}) T_{g_{2,p}^{-2}}f_{0,0}^{(1)}.
\end{align*}

Since $\deg (g_{2,p}^{-i} g_{k,p}) = \deg (g_k)$, $i\in \mathbb N$, $k=1,2,3$, it follows from equation \eqref{label 45} that

\begin{equation}
   \tilde{\mu}_{p} (Z_{0,0}^{(1)})= R(p,\beta) \tilde{\mu}_{p} (Z_{0,0}^{(1)}), \label{label 42}
\end{equation}

where

\begin{align*}
    R(p,\beta)
    & = p^{1-\beta}+p^{2-\beta}+p^{3-\beta}+p^{-\beta}-p^{1-2 \beta}-p^{2-2 \beta}-2 p^{3-2 \beta}-p^{4-2 \beta}-p^{5-2 \beta}+p^{3-3 \beta}\\
    & +p^{4-3 \beta}+p^{5-3 \beta}+p^{6-3 \beta}-p^{6-4 \beta}\\
    & = 1- p^6 \left(p^{-\beta}-1\right) \left(p^{-\beta}-p^{-1}  \right) \left(p^{-\beta}-p^{-2}\right) \left(p^{-\beta}-p^{-3}\right).
\end{align*}

We can repeat the same computations with the functions $T_{g}f_{0,0}^{(0)}$, $g\in GSp_{4}^+(\mathbb Q)$ instead. As a summary we get 

\begin{align*}
    \deg (g_{2,p}^{-1} g_{1,p}) T_{g_{2,p}^{-1} g_{1,p}}f_{0}^{(0)} 
    & = (1+p)f_{0}^{(0)} + (p^3+p^2) f_{1}^{(0)} \\
    \deg (g_{2,p}^{-1} g_{3,p}) T_{g_{2,p}^{-1} g_{3,p}}f_{0}^{(0)} 
    & = f_{0}^{(0)} + (p^2-1+p^3+p)f_{1}^{(0)} + p^4 f_{2}^{(0)}  \\
    \deg (g_{2,p}^{-1}) T_{g_{2,p}^{-1}}f_{0}^{(0)} 
    & = f_{1}^{(0)}\\
    \deg (g_{2,p}^{-2} g_{1,p}) T_{g_{2,p}^{-2} g_{1,p}}f_{0}^{(0)} 
    & =  (1+p)f_{1}^{(0)} + (p^3+p^2)f_{2}^{(0)} \\
    \deg (g_{2,p}^{-2}) T_{g_{2,p}^{-2}}f_{0,0}^{(0)} 
    & = f_{2}^{(0)}
\end{align*}
and 

\begin{equation}
      \tilde{\mu}_{p} (Z_{0}^{(0)})= R(p,\beta) \tilde{\mu}_{p} (Z_{0}^{(0)}). \label{label 43}
\end{equation}

Suppose that $\mu_p(Z_{0,0}^{(1)})\neq 0$. Since $\beta \notin \{1,2,3\}$, by equation  \eqref{label 42} we get that $\beta =0$. Hence

\begin{equation}
\tilde{\mu}_p (Z_{k_1,k_2})= \tilde{\mu}_p (Z_{k_1+2,k_2+2}),\quad k_1,k_2 \in \mathbb Z. \label{label 44}
\end{equation}

This is a contradiction since $\tilde{\mu}_p (MSp_{4}(\mathbb Z _p))< \infty$. This shows that $\tilde{\mu}_p(Z_{0,0}^{(1)})= 0$ and by induction we see that $\tilde{\mu}_p(Z_{k_1,k_2}^{(1)})= 0$ for $k_1,k_2 \in \mathbb N$. By equation \eqref{label 44} we see that this indeed holds for all  $k_1,k_2 \in \mathbb Z$. The same argument shows that $\tilde{\mu}_p (Z_{k}^{(0)})=0$ for $k\in \mathbb Z$. It follows from Lemma \ref{label 46} that the $\tilde{\mu}_p$-measure of the set of nonzero matrices $g\in MSp_{4}(\mathbb Q_p)$ with $\lambda(g)=0$ is zero.
\end{proof}

\begin{cor} \label{label 56}
We denote by $MSp_4 (\mathbb A_f)^*$ the set of elements $h \in MSp_{4}(\mathbb A_f)$ such that $\lambda(m_p)\neq 0$ for all primes $p$. Let $\mu_\beta$ be a measure in $\mathcal{E}_\beta$ and $\beta \notin \{0,1,2,3\}$. Then
$PGSp_{4}^+(\mathbb R)\times MSp_{4}(\mathbb A_f )^*$ is subset of full measure in $PGSp_{4}^+(\mathbb R)\times MSp_{4}(\mathbb A_f)$.
\end{cor}

\begin{proof}
Given a prime $p$, consider the restriction of $\mu_\beta$ to the set 

\begin{equation*}
   PGSp_{4}^+(\mathbb R) \times MSp_4 (\mathbb Q_p) \times \prod_{p\neq q} MSp_{4}(\mathbb Z_q),
\end{equation*}
and the measure $\mu_{\beta,p}$ on $PGSp_{4}^+(\mathbb R) \times MSp_4 (\mathbb Q_p)$ obtained from the projection on the first two coordinates. Since $\mu_{\beta} \in \mathcal{K}_\beta $ we have that $\mu_{\beta,p} (PGSp_{4}^+(\mathbb R) \times MSp_{4}(\mathbb Z _p)) < \infty$. Given any $g \in G_p$ and Borel $Z \in PGSp_{4}(\mathbb R)^+ \times  MSp_4{ ( \mathbb {Q}_p} )$ we get

\begin{equation*}
    \mu_{\beta, p}(g Z) = \mu_{\beta}(g ( Z \times \prod_{q\neq p} MSp_{4}(\mathbb Z _q))) = \lambda (g)^{-\beta} \mu_{\beta, p}(Z).
\end{equation*}

Thus the measure $\mu_{\beta,p}$ satisfies the conditions of Lemma \ref{label 47}, whence $PGSp_{4}^+(\mathbb R)\times GSp_{4}(\mathbb Q_p) $ is a subset of full $\mu_{\beta,p}$-measure. This shows that the $\mu_{\beta}$-measure of the set 

\begin{equation*}
    \{ PGSp_{4}^+(\mathbb R)\times MSp_{4}(\mathbb Q _p) \times \prod_{p\neq q} MSp_{4}(\mathbb Z_{q})\mid \lambda(h_p) = 0\}
\end{equation*}

is zero. Finally observe that the complement of $\mathbb H_2^{+}\times MSp_{4}^+(\mathbb A_f)$ is equal to  

\begin{equation*}
    \bigcup_{p\in \mathcal{P} } GSp_{4}^+(\mathbb Q)\{ PGSp_{4}^+(\mathbb R) \times MSp_{4}(\mathbb Q _p) \times \prod_{p\neq q} MSp_{4}(\mathbb Z_{q})\mid \lambda(h_p) = 0\},
\end{equation*}
which completes the proof.
\end{proof}

Given a prime number $p \in \mathcal{P}$ and $\beta \in \mathbb R ^{*}_+$, we consider the subsemigroup of $GSp_{4}(\mathbb Q)^+$ given by $$S_{2,p}=\bigcup_{l\geq 0} S_2(p^l).$$ Note that $\Gamma_2 \subseteq S_{2,p}$ for any prime number $p$ and the corresponding Dirichlet series (cf. Definition \ref{label 20}) is 

\begin{equation*}
    \zeta_{S_{2,p},\Gamma_2}(\beta) =\sum_{g\in \Gamma_2 \backslash S_{2,p} / \Gamma_2} \lambda(g)^{-\beta} \deg_{\Gamma_2}(g) = \sum_{l=0}^{\infty} p^{-\beta l} R_{\Gamma_2}(p^l).
\end{equation*}

\begin{prop} \label{label 36}
Suppose $\beta \in \mathbb R ^{*}_+$. Then $\zeta_{S_2,\Gamma_2}(\beta) < \infty$ if and only if $\beta > 3$. In this case we have that 

\begin{equation} \label{label 94}
    \zeta_{S_{2,p},\Gamma_2}(\beta) = \frac{1-p^{2-2\beta}}{(1-p^{3-\beta}) (1-p^{2-\beta}) (1-p^{1-\beta}) (1-p^{-\beta})}.
\end{equation}
\end{prop}

\begin{proof}
We combine the results from Proposition \ref{label 21} and Lemma \ref{label 22} to first compute $R_{\Gamma_2}(p^l)$:

\begin{align*}
    R_{\Gamma_2}(p^l)
        &  =\sum_{d_1|d_2|p^l} \deg_{\Gamma_1}(\pmqty{d_1 & 0 \\ 0     & d_2}) d_1^2d_2\\
        &  = \sum_{l_1 \leq l_2 \leq l} \deg_{\Gamma_1}(\pmqty{p^{l_1} & 0 \\     0 & p^{l_2}})p^{2l_1+l_2}\\
        &= \sum_{i=0}^{l}\sum_{k=0}^{l-i}\deg_{\Gamma_1}(\pmqty{p^i & 0 \\ 0 & p^{k+i}}) p^{2i}p^{i+k}\\
        &= \sum_{i=0}^{l} p^{3i}+ \sum_{i=0}^{l}\sum_{k=1}^{l-i} p^{k-1}(1+p)p^{3i+k}\\
        &= \frac{1}{(1-p)^2(1+p+p^2)}\Big(1-p^{2l}(p+p^2+p^3)+p^{3l}(p^2+p^4)\Big),
\end{align*}
since $\deg_{\Gamma_1}\pmqty{p^{l_1} & 0 \\ 0 & p^{l_2}}= p^{l_2-l_1-1}$ for $l_2\neq l_1$ \cite[Theorem 4.1, Chapter IV]{krieg1990hecke}. It is then clear that the series $ \zeta_{S_2,\Gamma_2}(\beta) $ converges if and only if $\beta > 3$ and

\begin{align*}
    \zeta_{S_{2,p},\Gamma_2}(\beta)
    & = \frac{1}{(1-p)^2(1+p+p^2)} \Big(\frac{(1-p)^2(1+p+p^2)p^{-\beta}(p+p^{\beta})}{(1-p^{3-\beta})(1-p^{-\beta})(1-p^{2-\beta})}\Big)\\
    & = \frac{1-p^{2-2\beta}}{(1-p^{3-\beta}) (1-p^{2-\beta}) (1-p^{1-\beta}) (1-p^{-\beta})},
\end{align*}
for $\beta > 3$ as desired.
\end{proof}

We are now ready to prove the main theorem of this section.

\begin{thm} \label{label 68}
The $GSp_4$-system $(\mathcal{A},(\sigma_t)_{t\in \mathbb R^{+}})$ does not admit a $\textrm{KMS}_\beta$ state for $0 <\beta < 3$ with $\beta \notin \{1,2\}$. 
\end{thm}

\begin{proof}
We put

\begin{equation}
    Y_p:= PGSp_{4}^+(\mathbb R) \times GSp_{4}(\mathbb Z_p)\times \prod_{q\neq p}MSp_{4}(\mathbb Z_q),
\end{equation}
and note that
\begin{equation}
    PGSp_{4}^+(\mathbb R) \times (MSp_{4}(\mathbb Z_{p}) \cap GSp_{4}(\mathbb Q_p)) \times \prod_{q\neq p} MSp_4(\mathbb Z_q) 
    = \bigcup_{s \in \Gamma_2 \backslash S_{2,p} / \Gamma_2} \Gamma_2 s Y_p. \label{label 48}
\end{equation}
It is easy to see that the sets $\Gamma_2 s Y_p$ are disjoint for $s\in \Gamma_2 \backslash S_{2,p} / \Gamma_2$ and the complement of their union in $\mathbb{H}_2^+ \times MSp_4{(\hZ)}$ is a subset of $(\mathbb H_{2}^+ \times MSp_{4}^+ (\mathbb A_f))^c$, whence by Lemma \ref{label 47} it has full measure for $\beta\notin \{1,2,3\}$. Let $\mu_\beta \in \mathcal{E}(K_\beta)$ and $\nu_\beta$ the measure on $\Gamma_2 \backslash \tilde{Y}$. Note that if $g\in G_p \cap GSp_{4}(\mathbb Z_p)$ then necessarily $\lambda(g)=1$ and $g\in MSp_{4}(\mathbb Z)$, hence $G_p \cap GSp_{4}(\mathbb Z_p) = \Gamma$. We can then apply \cite[Lemma 2.7]{laca2007phase} to the group $G_p$ (a simple calculation shows that any elementary matrix in $S_{2,p}$ is generated by $g_{1,p}, g_{2,p}$ and $g_{3,p}$) and the spaces $\tilde{X}$ and $Y_0=Y_p$. We obtain that for any $g\in \Gamma_2 \backslash S_{2,p} / \Gamma_2$, we have

\begin{equation*}
    \nu_\beta (\Gamma_2 \backslash \Gamma_2 g Y_p)= \lambda(g)^{-\beta} \deg_{\Gamma_2}(g) \nu_\beta (\Gamma_2 \backslash Y_p).
\end{equation*}
Observe that

\begin{align} 
     \nu_{\beta} (\Gamma_2 \backslash PGSp_{4}^+(\mathbb R)  \times MSp_{4}(\hZ))
    & = \sum_{g \in \Gamma_2 \backslash S_{2,p} / \Gamma_2} \nu_\beta(\Gamma_2 \backslash \Gamma_2 g Y_p) \nonumber  \\
    &= 
    \sum_{g \in \Gamma_2 \backslash S_{2,p} / \Gamma_2} \lambda(g)^{-\beta} \deg_{\Gamma_2}(g) \nu_{\beta}(\Gamma_2 \backslash Y_p). \nonumber 
\end{align}

Hence

\begin{equation} \label{label 87}
    1 =  \zeta_{S_{2,p},\Gamma_2}(\beta)
    \nu_{\beta}(\Gamma_2 \backslash Y_p),
\end{equation}

which is not possible for $\beta < 3$ by Proposition \ref{label 36}. This shows that there are no $\textmd{KMS}_{\beta}$-states for $\beta <3$ and $\beta \notin \{0,1,2\}.$

\end{proof}

\subsection{Low temperature region and Gibbs states}
\hfill\\

In this section, we study the equilibrium states of the $GSp_{4}$-system within the low temperature region. More specifically, we give an explicit construction of the extremal $\textmd{KMS}_{\beta}$ states for $\beta > 4$, which of course provides a complete description of the set of $\textmd{KMS}_\beta$ states.

\begin{thm} \label{label 75}

For $\beta>4$, the extremal \textmd{KMS} states of the \textmd{GSp$_4$}-system are given by the Gibbs states

\begin{equation}
    \phi_\beta(f)=\frac{\zeta(2\beta-2)\textmd{Tr} (\pi_{y}(f)e^{-\beta H_y})}{\zeta(\beta) \zeta(\beta-1)\zeta(\beta-2) \zeta(\beta-3)},
\end{equation}

where $y \in PGSp_{4}^+(\mathbb R) \times GSp_{4}(\hZ)$.
\end{thm}

\begin{proof}

Let $F$ be an arbitrary finite set of primes and denote by $G_F$ the group generated by $G_p$ for $p\in F$. We denote by $MSp_{4}^+(\mathbb Z)=GSp_{4}^+(\mathbb Q) \cap MSp_{4}(\mathbb Z)$ and we put

\begin{equation*}
    S_F:= \{m\in MSp_{4}^+(\mathbb Z) \mid \lambda(m)\in \mathbb N (F)\},
\end{equation*}
and 

\begin{equation*}
    Y_F = PGSp_{4}^+(\mathbb R) \times \prod_{p\in F} GSp_{4}(\mathbb Z_p) \times \prod_{q\notin F} MSp_{4}(\mathbb Z_q).
\end{equation*}

Similarly to the proof of Theorem \ref{label 68} (we replace $Y_p$ by $Y_F$ and $S_{2,p}$ by $S_{F}$), we get

\begin{equation*}
    1= \nu_\beta(\Gamma_2 \backslash Y_F) \zeta_{S_F,\Gamma_2}(\beta)= \nu_\beta(\Gamma_2 \backslash Y_F) \prod_{p\in F} \zeta_{S_{2,p},\Gamma_2} 
\end{equation*}

Note that $Y_F\subseteq Y_{F'}$ for $F' \subseteq F$ and the intersection of $Y_F$ over all finite primes is the set $PGSp_{4}^+(\mathbb R) \times GSp_{4}(\hZ)$. Hence for $\beta >4$, we get

\begin{equation*}
    \nu_\beta(\Gamma_2 \backslash (PGSp_{4}^+(\mathbb R) \times GSp_{4}(\hZ)) = \zeta_{MSp_{4}^+(\mathbb Z),\Gamma_2} (\beta)^{-1},
\end{equation*}
where
$$ \zeta_{MSp_{4}^+(\mathbb Z),\Gamma_2} (\beta)=\prod_{p\in \mathcal{P}} \zeta_{S_{2,p},\Gamma_2}= \frac{\zeta (\beta) \zeta (\beta-1) \zeta (\beta-2) \zeta (\beta-3)}{\zeta (2\beta-2)}.$$

On the other hand, the sets $\Gamma_2 s (PGSp_{4}^+(\mathbb R) \times GSp_{4}(\hZ))$ are disjoints for $s \in \Gamma_2 \backslash MSp_{4}^+(\mathbb Z) / \Gamma_2 $. We thus obtain 
\begin{align*}
    &\nu_\beta (\Gamma_2 \backslash MSp_{4}^+(\mathbb Z) (PGSp_{4}^+(\mathbb R) \times GSp_{4}(\hZ)))
    = \sum_{s \in \Gamma_2 \backslash MSp_{4}^+(\mathbb Z) / \Gamma_2 } \nu_{\beta} (\Gamma_2 \backslash \Gamma_2 s (PGSp_{4}^+(\mathbb R) \times GSp_{4}(\hZ))) \\
    &= \zeta_{MSp_{4}^+(\mathbb Z),\Gamma_2} (\beta) \nu_{\beta} (\Gamma_2 \backslash PGSp_{4}^+(\mathbb R) \times GSp_{4}(\hZ))=1,
\end{align*}

Hence $MSp_{4}^+(\mathbb Z) (PGS_{4}^+(\mathbb R) \times GSp_{4}(\hZ))$ has full measure in $PGS_{4}^+(\mathbb R) \times MSp_{4}(\hZ)$ and by Corollary \ref{label 3} the subset $PGSp_4^+(\mathbb R) \times GSp_4 (\mathbb A_f)$ has full measure in $PGSp_4^+(\mathbb R) \times MSp_4 (\mathbb A_f)$. Conversely, any probability $\Gamma_2$-invariant measure on $PGSp_4^+(\mathbb R) \times GSp_4(\hZ)$ extends \cite[Lemma 2.4]{laca2007phase} uniquely to a measure on $PGSp_4^+(\mathbb R) \times GSp_4(\mathbb A_f)$ satisfying condition \ref{label 67}. 

\par Suppose now that $\mu_\beta \in \mathcal{E}_{\beta}$ is a Dirac measure centered on $y \in PGSp_{4}^+(\mathbb R) \times GSP_{4}(\hZ)$. Then

\begin{align*}
    \phi (f)
    &= \sum_{s\in \Gamma_2 \backslash MSp_{4}^+(\mathbb Z) / \Gamma_2 }\int_{\Gamma_2 \backslash \Gamma_2 s (PGSp_{4}^+(\mathbb R) \times GSp_{4}(\hZ))} f(1,\omega)d\nu_\beta(\omega)\\
    &= \zeta_{MSp_{4}^+(\mathbb Z),\Gamma_2} \sum_{s\in \Gamma_2 \backslash MSp_{4}^+(\mathbb Z) / \Gamma_2 } \lambda(s)^{-\beta} \sum_{h\in \Gamma_2   \backslash \Gamma_2 s\Gamma_2} f(1,hy)\\
    &= \zeta_{MSp_{4}^+(\mathbb Z),\Gamma_2} (\beta)^{-1} \sum_{h\in \Gamma_2 \backslash MSp_4^+(\mathbb Z)} \lambda(h)^{-\beta} f(1,hy)\\
    &= \frac{ \textmd{Tr}(\pi_y(f) e^{-\beta H_y}) }{ \textmd{Tr}(e^{-\beta H_y})},
\end{align*}
since the operator $H_y$ is positive and $G_y=MSp_{4}^+(\mathbb Z)$ for $y\in GSp_{4}^+(\mathbb R)\times GSp_{4}(\hZ)$.

\end{proof}

\subsection{The critical region}

\hfill\\

\par We denote by $\hat{\mathcal{E}}_{\beta}$ the subset of right $GSp_{4}(\hZ)$-invariant measures in $\mathcal{E}_{\beta}$. The next proposition shows that this set is not empty for $3 \leq \beta < 4$.

\begin{prop} \label{label 71}
For each $\beta \in (3,4]$, the $\textmd{GSp}_4$-system $\{\mathcal{A},(\sigma_t)_{t\in \mathbb{R}}\}$ admits at least one $\textrm{KMS}_\beta$-state.
\end{prop}

\begin{proof}
We generalize the construction in \cite{laca2007phase}. By the correspondence in Proposition \ref{label 14}, it is enough to construct a measure $\mu_\beta$ on $\mathbb H_2 \times MSp_{4}(\mathbb A_f)$ such that $\mu_\beta \in \mathcal{E}_{\beta}$. For each prime $p$ and $3 < \beta \leq 4$, we consider the normalized Haar measure on $GSp_{4}(\mathbb Z_p)$ so that the total volume is $\zeta_{S_{2,p},\Gamma_2}(\beta)^{-1}$ (we denote this measure by $\textmd{meas}_{\beta,p}$). Observe that $GSp_{4}(\mathbb Q_p)= G_p GSp_{4}(\mathbb Z_p)$ and hence by \cite[lemma 2.4]{laca2007phase} we can uniquely extend this measure to a measure $\mu_{\beta,p}$ on $GSp_{4}(\mathbb Q_p)$ such that if $Z$ is a compact measurable subset in $GSp_{4}(\mathbb Q_p)$, then

\begin{equation*}
    \mu_{\beta,p}(Z) = \sum_{g\in G_p} |\lambda(g)|_p^{-\beta} \textmd{meas}_{\beta,p}(gZ \cap GSp_{4}(\mathbb Z_p)),
\end{equation*}
where $|a|_p$ denotes the $p$-adic valuation of $a$. Since $\mu_{\beta,p}(hZ)= |\lambda(h)|_p^{\beta}\mu_{\beta,p}(Z)$
for $g\in GSp_{4}(\mathbb Q_p)$, it is clear that $\mu_{\beta,p}$ is left $GSp_{4}(\mathbb Z_p)$-invariant. It is also right $GSp_{4}(\mathbb Z_p)$-invariant since the Haar measure $\textmd{meas}_{\beta,p}$ is right translation invariant. We extend the measure  $\mu_{\beta,p}(Z)$ to a measure on $MSp_{4}(\mathbb Q_p)$ by setting $\mu_{\beta,p}(Z):= \mu_{\beta,p}(Z \cap GSp_{4}(\mathbb Q_p))$ for Borel $Z \subseteq MSp_{4}(\mathbb Q_p)$. To extend this measure to $MSp_{4}(\mathbb A_f)$ we first check that $\mu_{\beta,p}(MSp_{4}(\mathbb Z_p))=1$. The proof of Lemma \eqref{label 47} applied to the space $MSp_{4}(\mathbb Q_p)$ shows that the set $MSp_{4}(\mathbb Z_p)\cap GSp_{4}(\mathbb Q_p)$ has full measure. Since this set is precisely $S_{2,p}GSp_{4}(\mathbb Z_p)$, then similarly to the calculation in the proof of Theorem \ref{label 68} we get

\begin{equation*}
    \mu_{\beta,p} (MSp_{4}(\mathbb Z_p))= \sum_{g \in \Gamma_2 \backslash S_{2,p} / \Gamma_2} \lambda(g)^{-\beta} \deg_{\Gamma_2}(g) \mu_{\beta,p}(GSp_{4}(\mathbb Z_p))=1.
\end{equation*}

We thus define a measure on $MSp_{4}(\mathbb A_f)$ by $\mu_{\beta,f}=\prod_{p\in \mathcal{P}} \mu_{\beta,p}$. Then for $g\in GSp_{4}(\mathbb Q)$ and measurable subset $Z\subseteq MSp_{4}(\mathbb A_f)$, we get

$$\mu_{\beta,f}(gZ)= \Big( \prod_{p\in \mathcal{P}} |\lambda(g)|_p^{\beta}\Big) \mu_{\beta,f}(Z)=\lambda(g)^{-\beta} \mu_{\beta,f}(Z).$$

If we denote by $\mu_{\beta,PGSp_{4}^+(\mathbb R)}$ the normalized Haar measure on $PGSp_{4}^+(\mathbb R)$ such that $\nu_{\beta,PGSp_{4}^+(\mathbb R)}$ is a probability measure on $\Gamma_2 \backslash PGSp_{4}^+(\mathbb R)$, then it is clear that the measure defined by $\mu_{\beta}:= \mu_{\beta,PGSp_{4}^+(\mathbb R)} \times \mu_{\beta,f}$ is an element of $\mathcal{K}_\beta$. By construction, $\mu_{\beta}$ is also right $GSp_{4}(\hZ)$-invariant; that is $\mu_{\beta} \in \hat{\mathcal{E}}_{\beta}$.
\end{proof}

Our next goal is to show that for $3 < \beta \leq 4$, the $\textmd{KMS}_\beta$ constructed in Proposition \ref{label 71} is the unique equilibrium state. We first recall the definition of an ergodic action.

\begin{definition}
If $\mu \in \mathcal{E}_{\beta} $, the action of $G$ on the measure space $(X,\mu)$ is ergodic if the following holds: If $A$ is any $G$-invariant Borel subset of $X$, then $\mu(A)=0$ or $\mu(A^c)=0$.
\end{definition}

\par Recall that if $W$ be a locally compact group then a character of $W$ is a continuous homomorphisms $\chi: W \rightarrow \mathbb{T}$.

\begin{lem} \label{label 53}
For $n\in \mathbb N$, we let $G_n=1+p^n \mathbb{Z}_p^{\times} \subseteq \mathbb{Z}_p$ and $\chi$ be any character of $\mathbb{Z}_p^{\times}$. Then $G_k \subseteq \ker(\chi) $ for some $k\in \mathbb{N}$.
\end{lem}

\begin{proof}

Consider the open subset of $\mathbb{T}$ given by 
\begin{equation}
    V=\{z\in \mathbb{T} \mid \textrm{Re} (z) > 0\}. \label{label 15}
\end{equation}
Observe first that the only subgroup of $V$ is $\{1\}$ and a fundamental system of neighborhood of the neutral element of $\mathbb{Z}_p^{\times}$ is given by the subgroups 
\begin{equation*}
    G_1 \supset G_2 \supset \dots \supset G_n \dots
\end{equation*}
Consider now the open subset of $\mathbb{T}$ given in \eqref{label 15}. Since the character $\chi$ is continuous, there exists an integer $k \geq 1$ such that $\chi(G_k) \subseteq V$. Now $\chi$ is homomorphism and therefore the subset $\chi(G_k)$ is a subgroup of $V$, that is $\chi(G_k)=1$.
\end{proof}

\begin{lem} \label{label 57}
Let $m$ be an integer and $B$ a finite set of prime numbers. Then the set 

\begin{equation}
    \{(n,\dots,n) \in \prod_{p\in B} \mathbb{Z}_p^{\times} \mid n\in \mathbb{Z}\,\, \textmd{and}\,\, (n,p)=1 \quad \forall p\in B\}
\end{equation}
is dense in $\prod_{p\in B} \mathbb{Z}_p^{\times}$.
\end{lem}

\begin{proof}
Any $a \in \mathbb{Z}_p^{\times}$ admits a $p$-adic expansion of the form

\begin{equation*}
    a= \sum_{i\geq 0} c_i p^i, \quad 0 < c_i< p.
\end{equation*}
Hence, given $\epsilon > 0$ and
$x=(x_1,\dots,x_{|B|}) \in \prod_{p\in B} \mathbb{Z}_p^{\times}$, we choose $a_k \in \mathbb{Z}$ and $e_k \in \mathbb{N}$ such that  
\begin{equation*}
|x_k-a_k|_{p_k} < \frac{\epsilon}{2}, \quad (a_k,p_k)=1, \quad   p_k^{-e_k} <\frac{\epsilon}{2}
\end{equation*}
for $k=1,\dots,|B|$. By the Chinese remainder theorem, the congruence system
\begin{equation*}
    n\equiv a_k \mod p_k ^{e_k},\quad 1 \leq k \leq |B|
\end{equation*}
has a solution $n \in \mathbb{Z}$. The condition $(a_k,p_k)$ implies that $(n,p_k)=1$ for all $1 \leq k \leq |B|$. Hence

\begin{equation*}
    |n-x_k|_{p_{k}} \leq |n-a_k|_{p_k} + |a_k -x_k|_{p_k} \leq p^{e_k} + \frac{\epsilon}{2} \leq \epsilon,
\end{equation*}
as desired.

\end{proof}

Given $m\in \mathbb N$, a Dirichlet character modulo $m$ is a function $\chi_m: \mathbb Z \rightarrow \mathbb C$ obtained by extending a character of $(\mathbb Z / m \mathbb Z )^{\times}$ to $0$ on $\mathbb Z / m \mathbb Z$ and lifted to $\mathbb Z$ by composition. The corresponding Dirichlet $L$-function is defined by 

\begin{equation*}
    L(s,\chi_m) = \sum_{n=1}^{\infty} \chi(n)/n^s,\quad s\in \mathbb C. 
\end{equation*}

\begin{thm} \label{label 58}

Let $\hat{\mu}_{\beta} \in \hat{\mathcal{E}}_{\beta}$ and $A=L^\infty (PGSp^+_{4}(\mathbb R) \times MSp_{4}(\mathbb A_f),\hat{\mu}_{\beta} )$. We have that 

\begin{equation}
    A^{PGSp_4^+(\mathbb R) \times GSp_{4}^+(\mathbb Q)}= \mathbb C.
\end{equation}
\end{thm}

\begin{proof}

It is enough to show that the action of $GSp_{4}^+(\mathbb Q)$ on $ ( MSp_4 {(\mathbb A _f), \hat{\mu}_{\beta,f} })$ (where $\hat{\mu}_{\beta,f}$ is the measure on $MSp_4 (\mathbb A _f)$ obtained by projecting onto the second factor) is ergodic. The strategy is similar to \cite{neshveyev2002ergodicity} and \cite{laca2007phase}. Since every $GSp_{4}^+(\mathbb Q)$-invariant subset is completely determined by its intersection with $MSp_{4}(\hZ)$, it is enough to show that the closed subspace

\begin{equation*}
    H= \{f\in L^2 (MSp_{4}(\hZ)), \d \hat{\mu}_{\beta,f} \mid V_m f =f, \forall m \in MSp_{4}^+(\mathbb Z)\},\quad (V_mf) (x) := f(mx),
\end{equation*}
consists of constant functions. Denote by $P$ the orthogonal projection onto $H$. Since every function in $H$ is $\Gamma_2$-invariant, it is enough to show that $P$ maps $\Gamma_2$-invariants functions to constants.

Consider the subspace 

\begin{equation*}
    H_F=\{f\in L^2 (MSp_{4}(\hZ)), \d \hat{\mu}_{\beta,f}) \mid V_s f =f, \forall s\in S_F\},
\end{equation*}
and denote by $P_F$ the orthogonal projection onto $H_F$. Consider the subset

\begin{equation*}
W_F=\prod_{p\in F}GSp_{4}(\mathbb Z _p) \times \prod_{q\notin F} MSp_{4}(\mathbb Z_q).
\end{equation*}

We obtain by Lemma \ref{label 49} and Corollary \ref{label 57} that the disjoint union $\cup_{s\in \Gamma_2 \backslash S_F / \Gamma_2} sW_F$ is a subset of full measure. Hence given any $\Gamma_2$-invariant function $f\in L^2 (MSp_{4}(\hZ)), \d \hat{\mu}_{\beta,f}) $, we deduce from \cite[Lemma 2.9]{laca2007phase} that  

\begin{equation}
    P_Ff= \zeta_{S_F,\Gamma_2}(\beta)^{-1} \sum_{s\in \Gamma_2 \backslash S_F /\Gamma_2} \lambda(s)^{-\beta} R_{\Gamma_2}(s) T_sf. \label{label 52}
\end{equation}

We fix a finite set of primes $B$ such that $B \cap F =\emptyset$ and consider the functions in $L^2(MSp_{4}(\hZ),\d \mu_{\beta,f})$ of the form

\begin{equation*}
    \chi_B(x) = \begin{cases}
                \chi(\lambda ( (x_p)_{p\in B} ) ) \quad &\text{if}\,\,\, x\in W_B\\    
                0 \quad &\text{otherwise} \end{cases},
\end{equation*}
where $\chi$ is a character of the compact abelian group $\prod_{p\in B} \mathbb Z_p^{\times}$.

We will first show that $P\chi_{B}(x)$ is constant a.e. This is easy to prove if the character $\chi$ is trivial. Indeed, in this cases $\chi_B=\mathbb 1_B$ and the projection formula \eqref{label 52} we get that $P_B\chi_B = \zeta_{S_B,\Gamma_2}(\beta)^{-1}$. Since $P=PP_B$ the result follows. We now consider the case where $\chi$ is non-trivial.
We first write $\chi = \prod_{p\in B} \chi_p$, where $\chi_p$ is a character of $\mathbb Z_p ^{\times}$ given by

\begin{equation*}
    \chi_{p}(a)= \chi(1,\dots,a,\dots,1).
\end{equation*}

Then by Lemma \ref{label 53} for each prime $p\in B$ there exists an integer $k_p \in \mathbb{N}$ such that $G_{k_p} \subseteq \ker (\chi_p)$. Let $m=\prod_{p\in B} p^{k_p}$ and define $\chi_m: \mathbb Z \rightarrow \mathbb C$ by

\begin{equation*}
    \chi_m (n) = \begin{cases}
                \chi((n,\dots,n)) \quad &\text{if}\,\,\, (n,m)=1\\
                0 \quad &\text{otherwise}
                \end{cases}
\end{equation*}

It is clear that $\chi_m$ is multiplicative and $\chi_m(n+m)=\chi _m(n)$ if $(n,m)\neq 1$. Suppose that $(n,m)=1$ so that $n\in \mathbb Z_p^{\times}$. Hence

\begin{equation*}
    \chi_m(n+m)=\prod_{p\in B} \chi_{p} (n+m)= \prod_{p\in B} \chi_{p} (n) \chi_p (1+mn^{-1})= \prod_{p\in B} \chi_{p}(n)=\chi_m (n), 
\end{equation*}

since $p^{k_p}$ divides $m$. Hence $\chi_m$ is a Dirichlet character modulo $m$. We claim that $\chi_m$ is nontrivial. Indeed, let $a\in \prod_{p\in B} \mathbb Z_p^{\times}$ such that $\chi(a)\neq 1$. By Lemma \ref{label 56} the set 

\begin{equation*}
    \{ (n,\dots,n) \in \prod_{p\in B} \mathbb Z_p^{\times} \mid n\in \mathbb Z\,\,\, \text{and}\,\,\, (n,m)=1\}
\end{equation*}

is dense in $\prod_{p\in B} \mathbb Z_{p}^{\times}$. Since $\chi$ is continuous there exists $n_0 \in \mathbb Z$ with $(n_0,m)=1$ and $1\neq \chi(n_0,\dots,n_0)=\chi_{m}(n_0)$ as desired.

Since $F\cap B = \emptyset$, for $s\in S_F$ we can write

\begin{equation*}
    (T_s \chi_B)(x)= \begin{cases}
                    \chi(\lambda((x_p)_{p\in B})) \chi_m(\lambda(s)) \quad &
                    \text{if}\,\,\, x\in W_B\\
                    0 \quad &\text{otherwise}.
                    \end{cases}
\end{equation*}

This allows us to obtain an explicit upper bound of the $L^2$-norm of $P\chi_B$ as follows. Since $P=PP_F$, from the projection formula \eqref{label 52} we get

\begin{align*}
\norm{P\chi_B} & = \norm{PP_F \chi_B} \leq \norm{P_F\chi_B} 
 \leq   \zeta_{S_F,\Gamma_2}(\beta)^{-1}
\abs{\sum_{s\in \Gamma \ S_F / \Gamma } \lambda(s)^{-\beta} R_{\Gamma_{2}}(s)\chi_{m}(\lambda(s))}\\
& =  \zeta_{S_F,\Gamma_2}(\beta)^{-1} \abs{\sum_{n\in \mathbb {N}(F)} n^{-\beta} R_{\Gamma_{2}} (n)\chi_m(n)}\\
&=  \zeta_{S_F,\Gamma_2}(\beta)^{-1} \abs{\prod_{p\in F} \sum_{l=0}^{\infty} p^{-l\beta} R_{\Gamma_2}(p^l)\chi_m(p^l)},
\end{align*}
since the function $R_{\Gamma_2}(n)$ is multiplicative by Lemma \ref{label 54}. As in the proof of Proposition \ref{label 36} we get

\begin{align*}
    \norm{P\chi_B}
    & \leq \zeta_{F_S,\Gamma_2}(\beta)^{-1} \prod_{p\in F} \abs{(\frac{1}{1-\chi_m(p)p^{-\beta}}-\frac{p+p^2+p^3}{1-\chi_m(p)p^{2-\beta}}+ \frac{p^2+p^4}{1-\chi_m(p)p^{3-\beta}}   )}\\
    & = \prod_{p\in F} \abs{\frac{ (1+p^{1-\beta} (\chi_m(p) -1)) - p^{2-2\beta}\chi_m(p)) (1-p^{3-\beta})(1-p^{2-\beta}) (1-p^{-\beta})}{(1-\chi_m(p)p^{3-\beta}) (1-\chi_m(p)p^{2-\beta}) (1-\chi_m(p)p^{1-\beta}) (1-\chi_m(p)p^{-\beta}) (1-p^{2-2\beta})   }  }\\
    & = \abs{\frac{(\sum_{n\in \mathbb{N}(F)}\frac{\chi_m(n)}{n^{\beta-3}}) (\sum_{n\in \mathbb{N}(F)}\frac{\chi_m(n)}{n^{\beta-2}})   (\sum_{n\in \mathbb{N}(F)}\frac{\chi_m(n)}{n^{\beta-1}})}
    {\zeta_{\mathbb{N}(F)}(\beta-3) \zeta_{\mathbb{N}(F)}(\beta-2) \zeta_{\mathbb{N}(F)}(\beta-1)} } \prod_{p\in F} \abs{\frac{1+\chi_m(p)p^{1-\beta} }{1+p^{1-\beta}} }\\
    & \leq \abs{\frac{(\sum_{n\in \mathbb{N}(F)}\frac{\chi_m(n)}{n^{\beta-3}}) (\sum_{n\in \mathbb{N}(F)}\frac{\chi_m(n)}{n^{\beta-2}})   (\sum_{n\in \mathbb{N}(F)}\frac{\chi_m(n)}{n^{\beta-1}})}
    {\zeta_{\mathbb{N}(F)}(\beta-3) \zeta_{\mathbb{N}(F)}(\beta-2) \zeta_{\mathbb{N}(F)}(\beta-1)} }\\
    & \leq 
    \frac{L(\beta-3,\chi_m) L(\beta-2,\chi_m) L(\beta -1,\chi_m)  }
    {  \zeta_{\mathbb N (F) } (\beta -3)   \zeta_{\mathbb N (F)} (\beta -2)   \zeta_{\mathbb N (F)} (\beta -1)}.
\end{align*}

Since the character $\chi_m$ is non-trivial and $3 < \beta \leq 4 $, it follows from \cite[Proposition 12, Chapter VI]{serre1970cours} that the the right hand side above can be made arbitrary small as $F \nearrow \mathcal{ P}$ (with $F \cap B =\emptyset$). This shows that $P\chi_B=0$, in particular $P\chi_B$ is again a constant function when $\chi$ is a nontrivial character.

Consider now the functions $F_B \in L^2(MSp_{4}(\hZ),\d \mu_{\beta,f})$ of the form

\begin{equation*}
    F_B(x)=\begin{cases}
            f((x_p)_{p\in B})\quad &\text{if}\,\,\, x\in W_B\\
            0 \quad &\text{otherwise}.
            \end{cases}
\end{equation*}
where $f\in L^2(\prod_{p\in B} GSp_{4}(\mathbb Z_p),\d \mu_B)$ and $\mu_{B}=( \pi_B)_* (\mu_{\beta,f})$, where $\pi_B$ is the projection $\pi_B: MSp_{4}(\hZ) \rightarrow \prod_{p\in B} MSp_4(\mathbb Z_p)$. We first show that it is enough to assume that $f$ is $\Gamma_2$-invariant. Indeed, when this is not the case, we denote by $Q$ the projection onto the space of $Sp_{4}(\hZ)$-invariant functions. Since $\hat{\mu}_{\beta,f}$ is $Sp_{4}(\hZ)$-invariant we have that

\begin{equation*}
QF_B(x)=  \int_{Sp_{4}(\hZ)} F_B (gx) \d g 
= \begin{cases}
\int_{\prod_{p\in B} Sp_{4}(\mathbb{Z}_p)} f(gx) \d g_B\quad & \text{if}\,\,\, x \in W_B \\
0 \quad &\text{otherwise},
\end{cases}
\end{equation*}

Since the Haar measure on $Sp_{4}(\hZ)$ is left-translation-invariant and $PQ=P$ (since $\Gamma_2$ is dense in $Sp_{4}(\hZ)$), this shows that WLOG we can assume that $f$ is $\Gamma_2$-invariant. Then by the density of $\Gamma_2$ in $\prod_{p\in B} Sp_{4}(\mathbb Z _p)$, we see that the function $f$ depends only on $\lambda((x_p)_{p\in B})$, i.e 

\begin{equation*}
    F_B(x)=\begin{cases}
            f'(\lambda(x_p)_{p\in B})\quad &\text{if}\,\,\, x\in W_B\\
            0 \quad &\text{otherwise},
            \end{cases}
\end{equation*}
where $f'$ is a square integrable function in $\prod_{p\in B} \mathbb Z^{\times}_{p}$ (with the natural pushforward measure). The linear span of $\chi$ ( where $\chi$ is a character of $\prod_{p\in B} \mathbb Z^{\times}_{p}$) form a dense subspace of such square integrable functions and we have shown that $P\chi_B$ is constant. It follows that $PF_B$ is constant.

We can easily verify that the adjoint of the the operator $V_s$ (where $s\in S_B$) is given by

\begin{equation*}
    (V^*_sh)(x)=\begin{cases}
                \lambda(s)^{\beta}h (s^{-1}x)\quad &\text{if}\,\,\, x\in sMSp_{4}(\hZ)\\
                0 &\quad \text{otherwise}.
                \end{cases}
\end{equation*}

Hence the map $\lambda(s)^{-\beta /2}V^*_s$ maps isometrically the functions of the form $F_B$ to functions in $L^2(MSp_{4}(\hZ),\d \mu_{\beta,f})$ of the form

\begin{equation} \label{label 55}
    G_B (x)= \begin{cases}
               g((x_p)_{p\in B}) \quad &\text{if}\,\,\, x\in sW_B\\
                0 &\quad \text{otherwise},
                \end{cases}
\end{equation}
where $g\in L^2(s \prod_{p\in B} GSp_{4}(\mathbb Z_p),\d \mu_B)$. It follows then from $P=PV_s^*$ that the projection $P$ maps every function of the form \eqref{label 55} to a constant function. The set $\cup_{s\in \Gamma_2 \backslash S_B / \Gamma_2} s W_B$ has full measure in $MSp_4(\hZ)$ and thus $P$ maps functions depending only on $(x_p)_{p\in B}$ to constants. Finally observe that as $B \nearrow \mathcal{P}$, the union of $L^2(\prod_{p\in B} MSp_{4}(\mathbb Z_p),\d \mu_{B})$ over all finite sets of primes is a dense subspace of  $L^2(\prod_{p\in B} MSp_{4}(\mathbb Z_p),\d \hat{\mu}_{\beta,f})$ (this follows from weak convergence of measures). Since $P$ maps this dense subspace to constant functions, this finishes the proof.

\end{proof}

\begin{lem}\label{label 33}
Let $f$ be a smooth function on $\Gamma \backslash PGSp_4^{+}(\mathbb{R})$ with a compact support and let $\Omega$ be any compact subset of $\Gamma \backslash PGSp_4^{+}(\mathbb{R})$. If we denote by $\d \mu $ the normalized Haar measure on $\Gamma \backslash PGSp_4^{+}(\mathbb{R})$, then for all $\epsilon > 0$ there exist $\kappa_1 (\epsilon) > 0$, $\kappa_2>0$ and $M_2>0$ (depending on $f$) such that the inequality
\begin{equation*}
    \frac{1}{R_{\Gamma_2}(m)} \sum_{a\in \Gamma_2 \backslash S_m/\Gamma_2} \abs{(T_af)(\tau)-\int_{\Gamma_2 \backslash PGSp_4^{+}(\mathbb{R})} f d\mu} \deg(a)
     \leq \kappa_1 m^{(2 \epsilon -1)} \prod_{i=1}^{i=l} (1+\kappa_2 p_{i}^{-1}) \frac{1}{(1-p^{-2 \epsilon -1}) ^{2}},
\end{equation*}
 holds for all $ \tau \in \Omega$ and every integer of the form $m=\prod_{i=1}^{i=l}p_i^{l_i}$ where $\min \{p_1,\dots, p_{l}\} > M_2$.
\end{lem}

\begin{proof}
Let $p$ be a given prime and $l\in \mathbb N$. The formula of $R_{\Gamma_2}(p^l)$ in the proof of Proposition \ref{label 36} gives

\begin{align*}
    (R_{\Gamma_2}(p^l) - p^{3l}) (1-p)^2(1+p+p^2) 
    & = 1 - p^{3l}-p^{1+2l} -p^{2+2l} - p^{2+2l} + p^{1+3l} + p^{3+3l}\\
    & \geq 1 + 3 p^{3l} -p^{2+2l}\\
    & \geq p^{3l} (1-p^{2-l}) \geq 0,\quad \forall l\geq 2
\end{align*}
A simple verification for the case $l=1$ shows that 

\begin{equation*}
    (R_{\Gamma_2}(p^l) - p^{3l}) (1-p)^2(1+p+p^2)  \geq 0,\quad \forall l\in \mathbb N,
\end{equation*}
that is 

\begin{equation}
    R_{\Gamma_2}(p^l) p^{-3l} \geq 1.
\end{equation}
We suppose now that $m\in \mathbb N$ has the form $m=\prod_{i=1}^{i=r} p_i^{l_i}$. Then by \cite[Theorem 1.7 and section 4.7]{clozel2001hecke} together with the calculations for $G=GSp_{2n}$, $n\geq 2$ in \cite[pages 22-23]{clozel2001hecke} and the degree formula in $\eqref{label 38}$ applied to $n=2$ (recall that the big $O$ depends only on $n$), we can find  $\kappa_1>0$, $\kappa_2>0$ and $M_2>0$ (depending on $f$) such that

\begin{align*}
 &\frac{1}{R_{\Gamma_2}(m)} \sum_{a\in \Gamma_2 \backslash S_m/\Gamma_2} \abs{(T_af)(\tau)-\int_{\Gamma_2 \backslash PGSp_4^{+}(\mathbb{R})} f d\mu} \deg(a)\\
& \leq  \frac{ \kappa_1}{R_{\Gamma_2}(m)} \sum_{ k_{ij} \leq m_{ij} \leq [l_i/2]} \prod_{i=1}^{i=r} (p^{2l_i-2k_{ij}-2m_{ij}})^{\epsilon -\frac{1}{2}} p_i^{3l_i-4k_{ij}-2m_{ij}} (1+O(p_{i}^{-1}))\\
& \leq \frac{ \kappa_1}{R_{\Gamma_2}(m) p^{-3l_i}} \prod_{i=1}^{i=l} \sum_{ k_{ij} \leq  m_{ij} \leq [l_i /2 ] } p^{l_i(2\epsilon -1)} p_i^{(-3-2\epsilon)k_{ij}+ (-1-2\epsilon) m_{ij}} (1+\kappa_2 p_i^{-1})\\
& \leq \kappa_1 m^{2\epsilon -1} 
\prod_{i=1}^{i=l} \sum_{ k_{ij} \leq  m_{ij} \leq [l_i /2 ] }  p^{-(3+2\epsilon)k_{ij} - (1+2\epsilon) m_{ij}} (1+\kappa_2 p_i^{-1})\\
& \leq \kappa_1 m^{2\epsilon -1} \prod_{i=1}^{i=l} (1+\kappa_2 p_i^{-1}) \sum_{m_{ij}=0}^{\infty} \sum_{k_{ij}=0}^{\infty}  p^{-(3+2\epsilon)k_{ij} - (1+2\epsilon) m_{ij}} \\
& \leq \kappa_1 m^{2\epsilon -1} \prod_{i=1}^{i=l} (1+\kappa_2 p_i^{-1}) \frac{1}{(1-p^{-1-2\epsilon})^2},
\end{align*}
holds for $\min \{p_1,\dots, p_{l}\} > M_2$ as desired.

\end{proof}

\begin{prop} \label{label 51}
Let $J$ be any nonempty finite set of prime numbers,  $3 < \beta \leq 4$, $f$ a smooth function on $\Gamma \backslash GSp_4^{+}(\mathbb{R})$ with a compact support and $\Omega$ any compact subset of $\Gamma \backslash GSp_4^{+}(\mathbb{R})$. Then for any $\delta >0$, there exists a sequence $\{F_n\}_{n\geq 1}$ consisting of finite set of prime numbers that are disjoint from $J$ and such that for all $\tau \in \Omega$ we have

\begin{equation*}
    \abs{(T_{F_n}f) (\tau) - \int_{\Gamma_2 \backslash PGSp_{4}^+ (\mathbb R)} f d\mu}< \delta.
\end{equation*}
\end{prop}

\begin{proof}

Given $3 < \beta \leq 4$, we fix $0 < \epsilon < \frac{\beta -3}{2} \leq \frac{1}{2}$ and choose $\kappa_1,\kappa_2$ and $M_2$ as in Lemma \ref{label 33}. Let $M_1 >0 $ be such that

\begin{equation}
    x^{\beta}(1-x^{2\epsilon -1} - \kappa_2 x^{2\epsilon -2}) > \kappa_2 \quad \forall x > M_1 \label{label 34},
\end{equation}

and we set $M:=\max\{M_1,M_2\}$. Let $F$ be a finite set of prime numbers with $F\cap J = \emptyset$ and $\min \{p\in F\} > M$. Then by Lemma \ref{label 33}, we have

\begin{align*}
&\abs{(T_Ff)(\tau) - \int_{\Gamma_2 \backslash PGSp_{4}^+ (\mathbb R) }  f d \mu }\\
    & \leq \kappa_1 \xi_{S_F,\Gamma}(\beta)^{-1} 
     ( \sum_{m\in \mathbb N (F)} m^{(2 \epsilon -1)} \prod_{i=1}^{i=l} (1+\kappa_2 p_{i}^{-1}) \frac{1}{(1-p^{-2 \epsilon -1}) ^{2}} R_{\Gamma}(m) m^{-\beta} )\\
    & \leq \kappa_1 \xi_{S_F,\Gamma} (\beta) ^{-1}
    (\prod_{p\in F} \sum_{l=0}^{\infty} p^{l(2\epsilon -1)} R_{\Gamma} (p^l) p^{-l\beta} (1+\kappa_2 p^{-1}) \frac{1}{(1-p^{-2\epsilon -1})^2})\\
    & \leq  \kappa_1 \xi_{S_F,\Gamma} (\beta) ^{-1} (\prod_{p\in F} (1+\kappa_2p^{-1}) \prod_{p\in F} \frac{1}{(1-p^{-2\epsilon -1})^2} \\
    & \prod_{p\in F} \frac{ 1+p^{2\epsilon -1} p^{\beta -1} }{1+p^{\beta -1}} \frac{\xi_{\mathbb N (F)} (\beta -2 -2\epsilon) \xi_{\mathbb N (F)} (\beta -1 -2\epsilon) \xi_{\mathbb N (F)} (\beta -2 \epsilon)}{ \xi_{\mathbb N (F)} (\beta -3) \xi_{\mathbb N (F)} (\beta -2) \zeta_{\mathbb N (F)} (\beta -1)}.
\end{align*}

Notice that by our choice of $M$ and $F$, Equation \eqref{label 34} gives

\begin{equation*}
    (1 + \kappa_2 p^{-1}) \frac{1+p^{2\epsilon-2+\beta}}{1+ p^{\beta-1}} < 1,\quad \forall p\in F.
\end{equation*}

Hence

\begin{align*}
      \abs{(T_Ff)(\tau) - \int_{\Gamma \backslash PGSp_{4}^+ (\mathbb R) }  f d \mu } \leq \kappa_1 \prod_{p\in F} \frac{\zeta_{\mathbb N (F)} (\beta -2 -2\epsilon) \zeta_{\mathbb N (F)} (\beta -1 -2\epsilon) \zeta_{\mathbb N (F)} (\beta -2 \epsilon)} { (1-p^{(-2\epsilon -1)})  \zeta_{\mathbb N (F)} (\beta -3) \zeta_{\mathbb N (F)} (\beta -2) \zeta_{\mathbb N (F)} (\beta -1)}.
\end{align*}

As $F \nearrow  \mathcal{P}$ with $F\cap J = \emptyset$, the right hand side can be made arbitrary small since $3 < \beta \leq 4$, $\beta - 2 -2\epsilon > 1$ and $\epsilon >0$, 
\end{proof}

\begin{thm} \label{label 59}
Let $\hat{\mu}_{\beta} \in \hat{\mathcal{E}}_{\beta}$ and $A=L^\infty (PGSp^+_{4}(\mathbb R) \times MSp_{4}(\mathbb A_f),\hat{\mu}_{\beta} )$. We have that 

\begin{equation}
    A^{GSp_4^+(\mathbb Q) \times GSp_{4}(\hZ)}= \mathbb C.
\end{equation}
\end{thm}

\begin{proof}
Consider the space $\mathcal{H}=L^2(GSp_{4}^+(\mathbb R) \times MSp_{4}(\hZ),\d \hat{\mu}_\beta)$. Observe that any $GSp_{4}^+(\mathbb Q) \times GSp_{4}(\hZ)$-invariant subset of $GSp^+_{4}(\mathbb R) \times MSp_{4}(\mathbb A_f)$ is completely determined by its intersection with
$GSp^+_{4}(\mathbb R) \times MSp_{4}(\mathbb \hZ)$, hence it is enough to show that any $MSp_{4}^+(\mathbb Z) \times GSp_{4}(\hZ)$-invariant function in $\mathcal{H}$ is constant. We denote by $H$ the closed subspace of $MSp_{4}^+(\mathbb Z) \times GSp_{4}(\hZ)$-invariant functions in $\mathcal{H}$ and denote by $P$ the orthogonal projection onto $H$. We will show that the image under $P$ of a dense subspace consists of constant functions. Given any non-empty finite sets of primes $F$ and $J$, we denote by $H_F$ the closed subspace of $S_F$-invariant functions in $\mathcal{H}$ and by $P_F$ the orthogonal projection onto $H_F$. Let $\mathcal{H}_J$ be the subspace of $\Gamma_2\times \prod_{p\in J} GSp_{4}(\mathbb Z_p)$-invariant functions depending only on $GSp_{4}^+(\mathbb R) \times \prod_{p\in J} MSp_{4}(\mathbb Z _p)$.

Recall that $S_J Y_J$ is a subset of full measure, whence by \cite[Lemma 2.9]{laca2007phase} given any function $f$ in $\mathcal{H}_J$ we get

\begin{equation}
    P_Jf = \zeta_{S_F,\Gamma_2}(\beta) \sum_{\Gamma_2 \backslash S_J /\Gamma_2} \lambda(s)^{-\beta}\deg_{\Gamma_2} (s) T_sf. \label{label 50}
\end{equation}

It follows that the value $P_Jf(\tau,m)$ depends only on $\tau \in GSP_{4}^{+}(\mathbb R)$. We can then write

\begin{equation*}
    P_Jf(x)=\begin{cases} 
            \tilde{f}(\tau) \quad &\text{if}\,\,\, x=(\tau,m) \in Y_J\\
            0 \quad & \text{otherwise}
            \end{cases}
\end{equation*}

We put $\tilde{f}_J:= P_J f$. Observe that since $\tilde{f}_J$ is $S_J$-invariant we get that $\tilde{f}$ is $\Gamma_2$-invariant and therefore can view it as a square integrable function on $\Gamma_2 \backslash PGSp_{4}^+(\mathbb R)$. We first suppose that $\tilde{f}$ is smooth with compact support $\Omega$. For $F \cap J = \emptyset$, the projection formula \eqref{label 50} gives

\begin{equation*}
    P_F\tilde{f}_J (x) = \begin{cases}
    T_F\tilde{f} (\tau) \quad & \text{if} \,\,\, x = (\tau, m) \in Y_J \\
    0 \quad & \text{otherwise}.
    \end{cases}
\end{equation*}

We put $(T_F \tilde{f})_J:= P_F\tilde{f}_J (x)$. Given $\epsilon >0$, by Proposition  \ref{label 51} there exists some finite set of primes $F$ disjoint from $J$ such that 

\begin{equation*}
    \abs{T_{F}f (\tau) -\int_{\Gamma_2 \backslash PGSp_{4}^{+}(\mathbb R)}\tilde{f}\d \mu } < \epsilon, \quad \forall \tau \in \Omega.
\end{equation*}

Since $PP_J=PP_{F}=P$, we get

\begin{align*}
    \norm{Pf - P \mathbb 1_{Y_J} \int_{\Gamma_2 \backslash PGSp_{4}^{+}(\mathbb R)} \tilde{f}\d \mu }_2
    & = \norm{PP_Jf - P \mathbb 1_{Y_J} \int_{\Gamma_2 \backslash PGSp_{4}^{+}(\mathbb R)} \tilde{f}\d \mu}_2\\
    & = \norm{PP_{F}P_Jf - P \mathbb 1_{Y_J} \int_{\Gamma_2 \backslash PGSp_{4}^{+}(\mathbb R)} \tilde{f}\d \mu}_2\\
    & \leq 
    \norm{P_{F}P_Jf -  \mathbb 1_{Y_J} \int_{\Gamma_2 \backslash PGSp_{4}^{+}(\mathbb R)} \tilde{f}\d \mu}_2\\
    & \leq \norm{(T_{F}\tilde{f} )_J -  \mathbb 1_{Y_J} \int_{\Gamma_2 \backslash PGSp_{4}^{+}(\mathbb R)} \tilde{f}\d \mu }_2 < \epsilon.
\end{align*}

Hence using the projection formula \ref{label 50} with $f=\mathbb 1 _{Y_J}$ we get

\begin{equation*}
    Pf=P \mathbb 1_{Y_J} \int_{\Gamma_2 \backslash PGSp_{4}^{+}(\mathbb R)} \tilde{f}\d \mu = P P_J 1_{Y_J} \int_{\Gamma_2 \backslash PGSp_{4}^{+}(\mathbb R)} \tilde{f}\d \mu= \zeta_{S_J,\Gamma_2}(\beta)^{-1}\int_{\Gamma_2 \backslash PGSp_{4}^{+}(\mathbb R)} \tilde{f}\d \mu,
\end{equation*}

which is constant. Any integrable functions on $\Gamma_2 \backslash PGSp_{4}^+(\mathbb R)$ can be approximated by a compactly supported smooth function and hence $Pf$ is constant for all $\Gamma_2$-invariant functions in $\mathcal{H}_J$. The results follows since the union of $\mathcal{H}_J$ over all finite set of primes is dense in the space of square integrable $\Gamma_2\times GSp_{4}(\mathbb \hZ)$-invariant functions.

\end{proof}

\begin{thm} \label{label 100}
For $3 <\beta \leq 4 $, The $GSp_4$-system admits a unique $\textmd{KMS}_\beta$ state.
\end{thm}

\begin{proof}
We will show that the set $\mathcal{E}_\beta$ consists of a single point. We first use \cite[Propositon 4.6]{laca2007phase} together with Theorem \ref{label 58} and Theorem \ref{label 59} to conclude that $A^{GSp_{4}^+(\mathbb{Q})}=\mathbb C$ where $A=L^\infty (PGSp_{4}^+(\mathbb R) \times MSp_{4}(\mathbb A_{f,\mathbb Q}), \hat{\mu}_\beta)$ for any $\hat{\mu}_\beta \in \hat{\mathcal{E}}_\beta$ and $3 < \beta \leq 4 $, in other words there exists a unique right $GSp_{4}(\hZ)$-invariant measure $\d \hat{\mu}_\beta$ in $\mathcal{E}_\beta$. Suppose now that $\upsilon_\beta$ is any other point of $\mathcal{E}_\beta$. Then the measure defined by $$\omega = \int_{GSp_{4}(\hZ)}g \cdot \upsilon_{\beta} \d g$$
is an element of $\mathcal{E}_\beta$ and by unicity we get that $\omega = \hat{\mu}_{\beta,f}$. Since the point $ \hat{\mu}_{\beta,f}$ is extremal we conclude that $ \hat{\mu}_{\beta,f}=\upsilon_{\beta}$. This completes the proof.
\end{proof}

\begin{remark}
We have studied the $GSp_{4}$-system in the region $\beta>0$ with  $\beta \notin \{1,2,3\}$. Let us now consider the cases where the inverse temperature is a pole of the Dirichlet series \eqref{label 94}. If $\beta=2$, it is possible to construct explicit measures $\mu_2 \in \mathcal{E}_2$. We consider the normalized Haar measure on $\mathbb A_{f,\mathbb Q}$ such that $\textmd{meas}(\hZ)=1$ and

\begin{equation}
    \textmd{meas}(aE)= \prod_{p}|a_p|_p\, \textmd{meas}(E). 
\end{equation}

for any $a\in \mathbb A_{f,\mathbb Q}^{\times}$ and measurable subset $E\subseteq \mathbb A_{f,\mathbb Q}$.  Let $ \mu_{f}$ be the product measure on $\mathbb A_{f,\mathbb Q}^{2}$. Since 

\begin{equation*}
    \mqty (  0 &  0 & 0 & 0 & x_1\\
            0 &   0 &  0 &  0& x_2 \\
            0 &   0 &  0 &  0 & x_3 \\
            0 &   0 &  0 &  0 & x_4 \\
                    ) \in MSp_{4}(\mathbb A_{f,\mathbb Q}), \quad x_1,\dots,x_4 \in  A_{f,\mathbb Q}
\end{equation*}

we may consider $ \mu_{f}$ as a measure on $MSp_{4}(A_{f,\mathbb Q})$ such that $\mu_f(MSp_{4}(\hZ))=1$. We claim that $\mu_2 = \mu_{\infty} \times \mu_{f} \in \mathcal{E}_2$. By construction it is enough to show that $\mu_2$ satisfies the scaling condition \eqref{label 67}. Given $g\in GSp_{4}^+(\mathbb Q)$, we can find $\gamma_1,\gamma_2\in \Gamma_2$ and a diagonal matrix $D \in GSp_{4}^+(\mathbb Q)$ such that $g=\gamma_1 D \gamma_2$. Since $\gamma \hZ ^4= \hZ^4$ for any $\gamma \in \Gamma_2$ and the Haar measure is translation invariant we conclude that $\mu_1 (g B)=\lambda^{-2}\mu(B)$ for any Borel subset of $MSp_4(\mathbb A_{f,\mathbb Q})$.

One would expect to use a similar construction for $\beta=1$ and $\beta =3$. However, since the only subspace of $\mathbb A_{f,\mathbb Q}^4$ stable under the action of $GSp_{4}(\mathbb Q)^+$ is $\mathbb A_{f,\mathbb Q}^4$ itself, this argument fails in the case $\beta =1$ or $\beta=3$. We conjecture that the $GSp_{4}$-system does not admit any $\textmd{KMS}_\beta$ state in these two cases.
\end{remark}

\begin{remark}
The results we prove in this paper completely classify the $\textmd{KMS}_\beta$ states on the Bost-Connes-Marcolli $GSp_{4}$-system. In fact, 
we will show that given $\beta >0$ with $\beta \notin \{1,2,3\}$, there exists a one-to-one correspondence between $\textmd{KMS}_\beta$ on the Connes-Marcolli $GSp_{4}$-system and the $GSp_{4}$ system $(\mathcal{A},\sigma_t)$. Recall the set 

$$F_Y=\{h \in MSp_4(\hZ) \mid \textmd{rank}_{\mathbb Q_p}(h_p) \leq 2 \,\,\, \text{for all} \,\,\, p \in \mathcal{P}\},$$

and consider the dynamical system $I=C^*_r (\Gamma_2 \backslash GSp_{4}^+(\mathbb Q) \boxtimes_{\Gamma_2} (\mathbb{H}_2^+\times F_Y))$. We claim that $I$ can not have any $\textmd{KMS}_{\beta}$ states. To see this, observe that any $KMS_{\beta}$ state on $I$ gives rise to a regular $\Gamma_2$-invariant measure $\mu_\beta$ on the space $\mathbb H_2^+ \times F_Y $ (note that unlike the case where the underlying space is an $r$-discreet principal groupoid, the support of this measure is not necessarily contained in $\mathbb H_2^+ \times F_Y $.) By the $\textmd{KMS}_{\beta}$ condition, this measure still satisfies the scaling property \ref{label 14}. Now since $\mathbb H_2^+= (\mathbb U^2 / \{\pm 1_{4}\})\backslash PGSp_{4}^+(\mathbb R) $, we can define a measure on $PGSp_{4}^+(\mathbb R) \times F_Y$ by the formula

\begin{equation*}
    \int_{PGSp_4^+(\mathbb R) \times F_Y} f(x) d\tilde{\mu}_{\beta}(x) = \int_{\mathbb H_2^+ \times F_Y} \Big(\int_{\mathbb U^2 /\{\pm 1_{4}\}} f(xg) \d g \Big)\d \mu_{\beta}(x).
\end{equation*}

The measure $\tilde{\mu}_{\beta}$ satisfies the condition \ref{label 14}. There is a canonical extension of this measure to a $\Gamma_2$-invariant measure $\mu_{\beta} \in \mathcal{E}_\beta $ on the space $PGp_{4}^+(\mathbb R) \times MSp_{4}(\mathbb A_f)$. This leads to a contradiction since the set $PGSp_{4}^+(\mathbb R) \times F_Y$ has measure zero by Corollary \ref{label 56}. This shows that the set $\mathbb H_2^+ \times F_Y$ can be ignored in the the analysis of $\textmd{KMS}_\beta$ states for $\beta >0$ and  $\beta \notin \{1,2,3\}$ and if we let $\tilde{Y}= Y \backslash F_Y$, it is clear that different $\textmd{KMS}_\beta$ state on $C_r^*( \Gamma_2 \backslash  GSp_{4}^+(\mathbb Q) \boxtimes_{\Gamma_2} \tilde{Y})$ give rise to different $\textmd{KMS}_\beta$ state on the $GSp_{4}$-system $(\mathcal{A},\sigma_t)$.
\end{remark}

 \hfill\\

\bibliographystyle{plain}
\bibliography{refs}
\end{document}